\documentclass[12pt]{article}

\usepackage[square,sort,comma,numbers]{natbib}
\usepackage{times} 
\usepackage{latexsym,epsfig,amssymb,amsmath,amsfonts,graphicx,amsthm,ifthen,pifont,comment,enumerate,setspace,multirow,color}
\usepackage{mhequ}
\usepackage{csquotes}

\usepackage{JASA_manu}
\usepackage{bm}
\usepackage[ruled,vlined]{algorithm2e}
\usepackage{epstopdf}
\usepackage{xcolor}
\usepackage{tikz}
\usepackage{hyperref}
\usetikzlibrary{bayesnet}
\usepackage{epstopdf}
\usepackage{multicol,booktabs,colortbl,tabularx}
\usepackage{JASA_manu}

\usepackage{listings}
\usepackage{xcolor}

\lstnewenvironment{R}{\lstset{
  language=R,columns=fixed,
  basicstyle={\footnotesize\ttfamily\color{black}},
  keywordstyle=\color{black},
  stringstyle=\color{black!50!black},
}}{}

\lstdefinestyle{Rstyle}
{
  language=R,columns=fixed,
  basicstyle={\footnotesize\ttfamily\color{black}},
  keywordstyle=\color{black},
  stringstyle=\color{black!50!black},
}

\def\begmat{\left(\begin{array}}\def\endmat{\end{array}\right)}

\def\bi{\begin{itemize}\setlength{\itemsep}{0pt}} \def\ei{\end{itemize}}

\def\bl{\begin{list}{\labelitemi}{\leftmargin=1em}\setlength{\itemsep}{-2.5pt}}  \def\el{\end{list}}
\def\bn{\begin{enumerate}} \def\en{\end{enumerate}}
\def\bt{\begin{table}[h]} \def\et{\end{table}}
\def\bc{\begin{center}} \def\ec{\end{center}}

\newtheorem{theorem}{Theorem}[section]

\theoremstyle{plain}

\theoremstyle{plain}

\theoremstyle{remark}

\theoremstyle{plain}

\newcommand \be{\begin{equs}}
\newcommand \ee{\end{equs}}

\begin{document}

\title{
\begin{Large}
\textbf{
Using Bayesian Statistics in Confirmatory Clinical Trials in the Regulatory Setting}
\end{Large}
}
\author{Se Yoon Lee\\
\begin{small}
seyoonlee.stat.math@gmail.com
\end{small}
\\
\begin{small}
Department of Statistics, Texas A\&M University, 3143 TAMU, College Station, TX, U.S.A.
\end{small}}
\date{}
\maketitle


\begin{abstract}
\noindent 
Bayesian statistics plays a pivotal role in advancing medical science by enabling healthcare companies, regulators, and stakeholders to assess the safety and efficacy of new treatments, interventions, and medical procedures. The Bayesian framework offers a unique advantage over the classical framework, especially when incorporating prior information into a new trial with quality external data, such as historical data or another source of co-data. In recent years, there has been a significant increase in regulatory submissions using Bayesian statistics due to its flexibility and ability to provide valuable insights for decision-making, addressing the modern complexity of clinical trials where frequentist trials are inadequate. For regulatory submissions, companies often need to consider the frequentist operating characteristics of the Bayesian analysis strategy, regardless of the design complexity. In particular, the focus is on the frequentist type I error rate and power for all realistic alternatives. This tutorial review aims to provide a comprehensive overview of the use of Bayesian statistics in sample size determination in the regulatory environment of clinical trials. Fundamental concepts of Bayesian sample size determination and illustrative examples are provided to serve as a valuable resource for researchers, clinicians, and statisticians seeking to develop more complex and innovative designs.
\baselineskip=17pt
\end{abstract}

\section{Introduction}\label{sec:Introduction}
Clinical trials are a critical cornerstone of modern healthcare, serving as the crucible in which medical innovations are tested, validated, and ultimately brought to patients \citep{friedman2015fundamentals}. Traditionally, these trials adhered to frequentist statistical methods since the 1940s, providing valuable insights into treatment effects. However, they may fall short in addressing the increasing complexity of modern clinical trials, such as personalized medicine \citep{zhou2008bayesian, fountzilas2022clinical}, innovative study designs \citep{carlin2022bayesian, wilson2021bayesian}, and the integration of real-world data into randomized controlled trials \citep{yue2018leveraging, sherman2016real, wang2019propensity}, among many other challenges \citep{woodcock2017master,moscicki2017drug,bhatt2016adaptive}.

These new challenges commonly necessitate innovative solutions. The US 21st Century Cures Act and the US Prescription Drug User Fee Act VI include provisions to advance the use of complex innovative trial designs \citep{CID2020FDAGuidance}. Generally, complex innovative trial designs have been considered to refer to complex adaptive, Bayesian, and other novel clinical trial designs, but there is no fixed definition because what is considered innovative or novel can change over time \citep{berry2006bayesian,jack2012bayesian,landau2013sample,CID2020FDAGuidance}. A common feature of many of these designs is the need for simulations rather than mathematical formulae to estimate trial operating characteristics. This highlights the growing embrace of complex innovative trial designs in regulatory submissions.

In this paper, our particular focus is on Bayesian methods. Guidance from the U.S. Food and Drug Administration (FDA) \citep{Bayesian2010FDAGuidance} defines Bayesian statistics as an approach for learning from evidence as it accumulates. Bayesian methods offer a robust and coherent probabilistic framework for incorporating prior knowledge, continuously updating beliefs as new data emerge, and quantifying uncertainty in the parameters of interest or outcomes for future patients \citep{spiegelhalter2004bayesian}. The Bayesian approach aligns well with the iterative and adaptive nature of clinical decision-making, offering opportunities to maximize clinical trial efficiency, especially in cases where data are sparse or costly to collect.

The past ten years have seen notable demonstrations of Bayesian statistics addressing various types of modern complexities in clinical trial designs. For example, Bayesian group sequential designs are increasingly used for seamless modifications in trial design and sample size to expedite the development process of drugs or medical devices, while potentially leveraging external resources \citep{wilber2010comparison, gsponer2014practical, bohm2020efficacy, schmidli2020beyond,schmidli2007bayesian}. One recent example is the COVID-19 vaccine trial, which includes four Bayesian interim analyses with the option for early stopping to declare vaccine efficacy before the planned trial end \citep{polack2020safety}. Other instances where Bayesian approaches have demonstrated their promise are umbrella, basket, or platform trials under master protocols \citep{Master2022FDAGuidance}. In these cases, Bayesian adaptive approaches facilitate the evaluation of multiple therapies in a single disease, a single therapy in multiple diseases, or multiple therapies in multiple diseases \citep{berry2016response, chu2018bayesian, hirakawa2018master,hobbs2018bayesian, dodd2016design,quintana2023design,alexander2018adaptive,i2022clinical}. Moreover, Bayesian approaches provide an effective means to integrate multiple sources of evidence, a particularly valuable aspect in the development of pediatric drugs or medical devices where small sample sizes can impede traditional frequentist approaches \citep{wang2022bayesian,psioda2020bayesian,gamalo2017statistical}. In such cases, Bayesian borrowing techniques enable the integration of historical data from previously completed trials, real-world data from registries, and expert opinion from published resources. This integration provides a more comprehensive and probabilistic framework for information borrowing across different sub-populations \citep{ibrahim2000power,richeldi2022trial,muller2023bayesian,Pediatric2016FDAGuidance}.

The basic tenets of good trial design are consistent for both Bayesian and frequentist trials. It is crucial to note that, for regulatory submissions, sponsors (typically pharmaceutical or medical device companies) should adhere to the principles of good clinical trial design and execution, including minimizing bias, as outlined in regulatory guidances \citep{ADAPTIVEMD2016FDAGuidance, ADAPTIVEDRUG2019FDAGuidance, Bayesian2010FDAGuidance}.

In this paper, we provide a pedagogical understanding of Bayesian designs by elucidating key concepts and methodologies through illustrative examples. For the simplicity of explanation, we apply Bayesian methods to construct single-stage design, two-stage design, and multiple-arm parallel design for single-arm trials, but the illustrated key design principles can be generalized to multiple arms trials. Specifically, our focus in this tutorial is on Bayesian sample size determination, which is most useful in confirmatory clinical trials, including late phase II or III trials in the drug development process or pivotal trials in the medical device development process. We highlight the advantages of Bayesian designs, address potential challenges, examine their alignment with evolving regulatory science, and ultimately provide insights into the use of Bayesian statistics for regulatory submissions.

This paper is organized as follows. We explain a simulation-based approach to determine the sample size of a Bayesian design in Section \ref{sec:Sizing a Bayesian Trial}, which is consistently used throughout the paper. The specification of the prior distribution for Bayesian submission is discussed in Section \ref{sec:Specification of Prior Distributions}. Two important Bayesian decision rules, namely, the posterior probability approach and the predictive probability approach, are illustrated in Sections \ref{sec:Bayesian Decision Rule - Posterior Probability Approach} and \ref{sec:Bayesian Decision Rule - Predictive Probability Approach}, respectively. A Bayesian approach using Bayesian hierarchical modeling to address the multiplicity issue is illustrated in Section \ref{sec:Multiplicity Adjustments}. In Section \ref{sec:Historical Data Information Borrowing}, a Bayesian information borrowing technique is illustrated with some regulatory considerations. We conclude the paper with a discussion in Section \ref{sec:Discussion}.



\section{Sizing a Bayesian Trial}\label{sec:Sizing a Bayesian Trial}
\subsection{A Simulation Principle of Bayesian Sample Size Determination}\label{subsec:Simulation-based Bayesian Sample Size Determination}
Although practical and ethical issues need to be considered, one's initial reasoning when determining the trial size should focus on the scientific requirements (Pocock, 2013). Scientific requirements refer to the specific criteria, conditions, and standards that must be met in the design, conduct, and reporting of scientific research to ensure the validity, reliability, and integrity of the findings. Much like frequentist approaches for determining the sample size of the study (Chow, 2017), its Bayesian counterpart also proceeds by first defining a success criterion to align with the primary objective of the trial. Subsequently, the number of subjects is determined to provide a reliable answer to the questions addressed within regulatory settings.



In the literature, various studies have explored the sizing of Bayesian trials \cite{inoue2005relationship, katsis1999bayesian, joseph1995sample, rubin1998sample, joseph1995some, lindley1997choice, wang2002simulation}. Among these, the simulation-based method proposed by \cite{wang2002simulation} stands out as popular, and it was further explored by \citep{psioda2018bayesian, chen2011bayesian} for practical applications. This method is widely used by many healthcare practitioners, including design statisticians at companies or universities, for its practical applicability in a broad range of Bayesian designs. Furthermore, this method, with a particular prior setting, is well-suited for the regulatory environment, where the evaluation of the frequentist operating characteristics of the Bayesian design is critical. This will be discussed in Subsection \ref{subsec:Calibration of Bayesian Trial Design}.


In this section, we outline the framework of the authors' work \cite{wang2002simulation}. Similar to the notation in Reference \citep{lehmann1986testing} assume that the endpoint has probability density function $f(y|\theta)$, where the $\theta \in \Theta$ represents the parameter of main interest. The hypotheses to be investigated are the null and alternative hypotheses, 
\begin{align}
\label{eq:hypothesis_to_test}
\mathcal{H}_{0}: \theta \in \Theta_{0} \text{ versus } \mathcal{H}_{a}: \theta \in \Theta_{a},
\end{align}
where $\Theta_{0}$ and $\Theta_{a}$ represent the disjoint parameter spaces for the null and alternative hypotheses, respectively. $\Theta = \Theta_{0} \cup \Theta_{a}$ denotes the entire parameter space. Suppose that the objective of the study is to evaluate the efficacy of a new drug, achieved by rejecting the null hypothesis. Let $\textbf{y}_{N} = (y_{1},\cdots,y_{N})^{\top}$ denotes a set of $N$ outcomes such that $y_{i}$ ($i=1,\cdots,N$) is identically and independently distributed according a distribution $f(y|\theta)$. 

Throughout the paper, we assume that the parameter space $\Theta$ is a subset of real numbers. The range of the parameter space $\Theta$ is determined by the type of outcomes. For example, for continuous outcomes $y$, the distribution $f(y|\theta)$ may be a normal distribution, where the parameter space is the set of real numbers, $\Theta = \mathbb{R}$; and for binary outcomes, the distribution $f(y|\theta)$ is the Bernoulli distribution, where the parameter space is the set of fractional numbers, $\Theta = [0,1]$. In this formulation, typically, the hypotheses (\ref{eq:hypothesis_to_test}) are one-sided; for example, $\mathcal{H}_0: \theta \leq \theta_0$ versus $\mathcal{H}_a: \theta > \theta_0$ or $\mathcal{H}_0: \theta \geq \theta_0$ versus $\mathcal{H}_a: \theta < \theta_0$. Throughout the paper, when we denote hypotheses in the abstract form (\ref{eq:hypothesis_to_test}), it is considered a one-sided superiority test for the coherency of the paper. The logic explained in this paper can be generalized to a form of a two-sided test, non-inferiority test, or equivalence test in a similar manner, but discussion on these forms is out of scope for this paper.

The simulation-based approach incorporates two essential components: the `sampling prior' $\pi_{s}(\theta)$ and the `fitting prior' $\pi_{f}(\theta)$. The sampling prior is utilized to generate observations $\textbf{y}_{N}$ by considering the scenario of `what if the parameter $\theta$ is likely to be within a specified portion of the parameter space?' The fitting prior is employed to fit the model once the data $\textbf{y}_{N}$ has been obtained upon completion of the study. We note that the sampling prior should be a proper distribution, while the fitting prior does not need to be proper as long as the resulting posterior, $\pi(\theta|\textbf{y}_{N}) \propto f(\textbf{y}_{N}|\theta)\cdot \pi_{f}(\theta)$, is proper. We also note that the sampling prior is a unique Bayesian concept adopted in the simulation-based approach, whereas the fitting prior refers to the prior distributions used in the daily work of Bayesian data analyses \citep{gelman1995bayesian}, not confined to the context of sample size determination.

In the following, we illustrate how to calculate the Bayesian test statistic, denoted as $T(\textbf{y}_N)$, using the posterior probability approach by using a sampling prior and a fitting prior. (Detail of the posterior probability approach is illustrated in Section \ref{sec:Bayesian Decision Rule - Posterior Probability Approach}) First, one generates a value of parameter of interest $\theta$ from the sampling prior $\pi_{s}(\theta)$, and then generates the outcome vector $\textbf{y}_{N} = (y_{1},\cdots,y_{N})^{\top}$ based on that $\theta$. This process produces $N$ outcomes $\textbf{y}_{N}$ from its prior predictive distribution (also called, marginal likelihood function)
\begin{align}
\label{eq:prior_pred_dist}
\textbf{y}_{N} \sim f_{s}(\textbf{y}_{N}) &= \int f(\textbf{y}_{N}|\theta) \pi_{s}(\theta) d\theta.
\end{align}
After that, one calculates the posterior distribution of $\theta$ given the data $\textbf{y}_{N}$, which is 
\begin{align}
\label{eq:posterior_dist}
\pi_{f}(\theta|\textbf{y}_{N}) &= \frac{f(\textbf{y}_{N}|\theta) \pi_{f}(\theta)}{\int f(\textbf{y}_{N}|\theta) \pi_{f}(\theta) d\theta} .
\end{align}

Eventually, a measure of evidence to reject the null hypothesis is summarized by the Bayesian test statistics, the posterior probability of the alternative hypothesis being true given the observations $\textbf{y}_{N}$, which is $$T(\textbf{y}_N) = \mathbb{P}_f[\theta \in \Theta_a | \textbf{y}_N] = \int \textbf{1}\{\theta \in \Theta_a\} \pi_{f}(\theta|\textbf{y}_{N}) d \theta,$$ where the indicator function $\textbf{1}\{A\}$ is 1 if A is true and 0 otherwise. A typical success criterion takes the form of
\begin{align}
\label{eq:Bayesian_test_procedure}
\text{Study Sucess} = \textbf{1}\{ T(\textbf{y}_N) > \lambda \} =\textbf{1}\{\mathbb{P}_{f}[\theta \in \Theta_{a} | \textbf{y}_{N}] > \lambda \},
\end{align}
where $\lambda \in [0,1]$ is a pre-specified threshold value.

At this point, we introduce a key quantity to measure the expected behavior of the Bayesian test statistics $T(\textbf{y}_N)$—the probability of study success based on the Bayesian testing procedure—by considering the idea of repeated sampling of the outcomes $\textbf{y}_{N} \sim f_{s}(\textbf{y}_{N})$:
\begin{align}
\label{eq:Bayesian_Power_Function}
\beta_{\Theta}^{(N)} &= \mathbb{P}_{s}[T(\textbf{y}_{N}) > \lambda | \textbf{y}_{N} \sim f_{s}(\textbf{y}_{N})]
=\int \textbf{1}\{\mathbb{P}_{f}[\theta \in \Theta_{a} | \textbf{y}_{N}] > \lambda \} f_{s}(\textbf{y}_{N}) d\textbf{y}_{N}.
\end{align}
In the notation $\beta_{\Theta}^{(N)}$ (\ref{eq:Bayesian_Power_Function}), the superscript `$N$' indicates the dependence on the sample size $N$, and the subscript `$\Theta$' represents the support of the sampling prior $\pi_{s}(\theta)$. Note that in the equation (\ref{eq:Bayesian_Power_Function}), the probability inside of $\textbf{1}\{A\}$ (that is, $\mathbb{P}_{f}[\cdot]$) is computed with respect to the posterior distribution $\pi_{f}(\theta|\textbf{y}_{N})$ (\ref{eq:posterior_dist}) under the fitting prior, while the probability outside (that is, $\mathbb{P}_{s}[\cdot]$) are taken with respect to the marginal distribution $f_{s}(\textbf{y}_{N})$ (\ref{eq:prior_pred_dist}) under the sampling prior. Note that the value $\beta_{\Theta}^{(N)}$ (\ref{eq:Bayesian_Power_Function}) also depends on the choice of the threshold ($\lambda$), the parameter spaces corresponding to the null and alternative hypothesis ($\Theta_{0}$ and $\Theta_{a}$), and the sampling and fitting priors ($\pi_{s}(\theta)$ and $\pi_{f}(\theta)$).

For many cases, Monte Carlo simulation is employed to approximate the value of $\beta_{\Theta}^{(N)}$ (\ref{eq:Bayesian_Power_Function}) because it is not expressed as a closed-form formula:
\begin{align*}
\hat{\beta}_{\Theta}^{(N)} \approx \frac{1}{R}\sum_{r=1}^{R}
\textbf{1}\{\mathbb{P}_{f}[\theta \in \Theta_{a} | \textbf{y}_{N}^{(r)}] > \lambda \}, \quad  \textbf{y}_{N}^{(r)}\sim f_{s}(\textbf{y}_{N}), \quad (r = 1,\cdots,R),
\end{align*}
where $R$ is the number of simulated datasets. When Monte Carlo simulation is used for regulatory submission in a Bayesian design to estimate the expected behavior of the Bayesian test statistics $T(\textbf{y}_N)$, typically, one uses $R=10,000$ or $100,000$ and also reports a 95\% confidence interval for $\beta_{\Theta}^{(N)}$ to describe the precision of the approximation. Often, for complex designs, computing the Bayesian test statistic $T(\textbf{y}_{N}) = \mathbb{P}_{f}[\theta \in \Theta_{a} | \textbf{y}_{N}]$ itself requires the use of Markov Chain Monte Carlo (MCMC) sampling techniques, such as the Gibbs sampler or Metropolis-Hastings algorithm \citep{gamerman2006markov, andrieu2003introduction, lee2022gibbs}. In such cases, a nested simulation technique is employed to approximate $\beta_{\Theta}^{(N)}$ (\ref{eq:Bayesian_Power_Function}) (Algorithm \ref{alg:nested simulation}). It is important to note that when MCMC techniques are used, regulators recommend sponsors check the convergence of the Markov chain to the posterior distribution \citep{Bayesian2010FDAGuidance}, using various techniques to diagnose nonconvergence \citep{gelman1995bayesian, gamerman2006markov}.

Now, we are ready to apply the above concept to Bayesian sample size determination. We consider two different populations from which the random sample of $N$ observations $\textbf{y}_{N}$ may have been drawn, with one population corresponding to the null parameter space $\Theta_{0}$ and another population corresponding to the alternative parameter space $\Theta_{a}$—similar to Neyman \& Pearson’s approach (based on hypothesis testing and type I and II error rates) \citep{neyman1933ix}. 

This can be achieved by separately considering two scenarios: `what if the parameter $\theta$ is likely to be within a specified portion of the null parameter space?' and `what if the parameter $\theta$ is likely to be within a specified portion of the alternative parameter space?' Following notations from \citep{chen2011bayesian}, let $\bar{\Theta}_{0}$ and $\bar{\Theta}_{a}$ denote the closures of $\Theta_{0}$ and $\Theta_{a}$, respectively. In this formulation, the null sampling prior $\pi_{s0}(\theta)$ is the distribution supported on the boundary $\Theta_{B} = \bar{\Theta}_{0} \cap \bar{\Theta}_{a}$, and the alternative sampling prior $\pi_{s1}(\theta)$ is the distribution supported on the set $\Theta_{a}^{*}\subset \Theta_{a}$. For a one-sided test, such as $\mathcal{H}_{0}: \theta \leq \theta_{0}$ versus $\mathcal{H}_{a}: \theta > \theta_{0}$, one may choose the null sampling prior $\pi_{s0}(\theta)$ as a point-mass distribution at $\theta_{0}$, and the alternative sampling prior $\pi_{s1}(\theta)$ as a distribution supported on $\Theta_{a}^{*}\subset(\theta_{0},\infty)$. 

Eventually, for a given $\alpha > 0$ and $\beta > 0$, the Bayesian sample size is the value
\begin{align}
\label{eq:optimal_Bayesian_Sample_Size}
N = \text{max}\left(\text{min}\{N : \beta_{\Theta_{B}}^{(N)} \leq \alpha \},
 \text{min}\{N : \beta_{\Theta_{a}^{*}}^{(N)} \geq 1 - \beta \} \right),
\end{align}
where $\beta_{\Theta_{B}}^{(N)}$ and $\beta_{\Theta_{a}^{*}}^{(N)}$ are given in (\ref{eq:Bayesian_Power_Function}) corresponding to $\pi_{s}(\theta) = \pi_{s0}(\theta)$ and $\pi_{s}(\theta) = \pi_{s1}(\theta)$, respectively. The values of $\beta_{\Theta_{B}}^{(N)}$ and $\beta_{\Theta_{a}^{*}}^{(N)}$ are referred to as the Bayesian type I error and power, while $1 - \beta_{\Theta_{a}^{*}}^{(N)}$ is referred to as the Bayesian type II error. The sample size $N$ satisfying the condition $\beta_{\Theta_{B}}^{(N)} \leq \alpha $ meets the Bayesian type I error requirement. Similarly, the sample size $N$ satisfying the condition $\beta_{\Theta_{a}^{*}}^{(N)} \geq 1 - \beta$ meets the Bayesian Power requirement. Eventually, the selected sample size $N$ (\ref{eq:optimal_Bayesian_Sample_Size}) is the minimum value that simultaneously satisfies the Bayesian type I error and power requirement. Typical values for $\alpha$ are 0.025 for a one-sided test and 0.05 for a two-sided test, and $\beta$ is typically set to 0.1 or 0.2 regardless of the direction of the alternative hypothesis \citep{Bayesian2010FDAGuidance}.

\begin{figure}[h!]
\centering
\includegraphics[width=\textwidth]{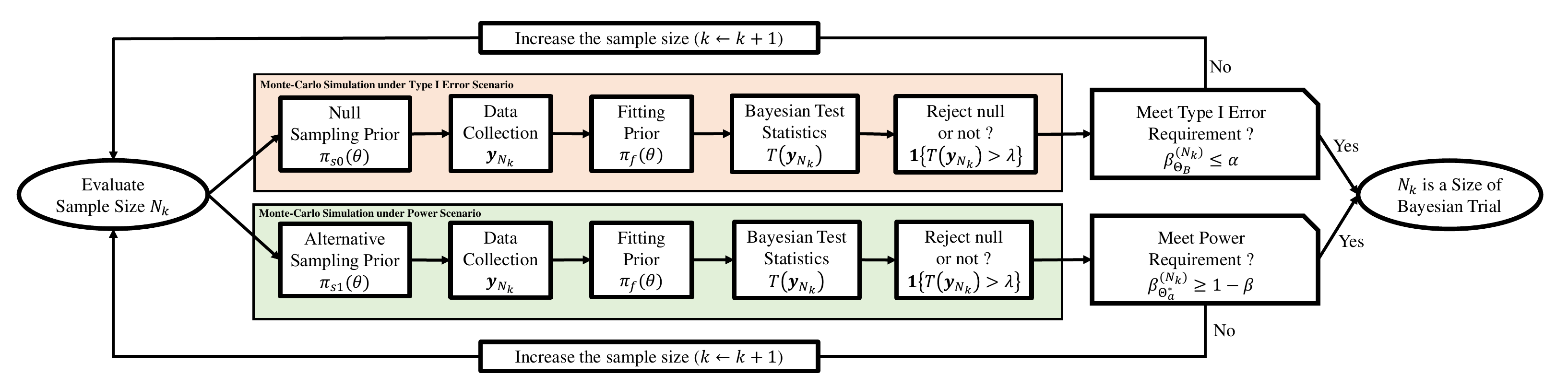}
\caption{\baselineskip=12pt Flow chart of Bayesian sample size determination within the collection of possible sizes of Bayesian trial $\mathcal{N} = \{N_{k}; k=1,\cdots,K, N_{k} < N_{k+1}\}$.}
\label{fig:BSSD_Gelfand}
\end{figure}

Figure \ref{fig:BSSD_Gelfand} provides a flowchart illustrating the process of Bayesian sample size determination. We explain the practical algorithm for selecting an optimal Bayesian sample size $N$ (\ref{eq:optimal_Bayesian_Sample_Size}), subject to the maximum sample size $N_{\text{max}}$—typically chosen under budgetary limits. To begin, we consider a set of $K$ candidate sample sizes, denoted as $\mathcal{N} = \{N_{k}; k=1,\ldots,K, N_{k} < N_{k+1}, , N_{K}=N_{\text{max}}\}$. Often, one may include the frequentist sample size as a reference.

The process commences with the evaluation of the smallest sample size, $N_{1}$, checking whether it meets the Bayesian type I error and power requirements, i.e., $\beta_{\Theta_{B}}^{(N_{1})} \leq \alpha$ and $\beta_{\Theta_{a}^{*}}^{(N_{1})} \geq 1 - \beta$. To that end, we independently generate $N_{1}$ outcomes, $\textbf{y}_{N_{1}}$, from the marginal distributions $f_{s0}(\textbf{y}_{N_{1}})$ and $f_{s1}(\textbf{y}_{N_{1}})$, which are based on the null and alternative sampling priors $\pi_{s0}(\theta)$ and $\pi_{s1}(\theta)$, respectively. The data drawn in this manner corresponds to the type I error and power scenarios, respectively. Subsequently, we independently compute the Bayesian test statistics, $T(\textbf{y}_{N_{1}})$, using the common fitting prior $\pi_{f}(\theta)$, and record the testing results, whether it rejects the null hypothesis or not, $\textbf{1}\{T(\textbf{y}_{N_{1}}) > \lambda \}$ (\ref{eq:Bayesian_test_procedure}) for each scenario. By repeating this procedure $R$ times (for example, $R = 10,000$), we can estimate the expected behaviors of the Bayesian test statistics $\beta_{\Theta_{B}}^{(N_{1})}$ and $\beta_{\Theta_{a}^{*}}^{(N_{1})}$ through Monte-Carlo approximation and evaluate whether the size $N_{1}$ meets both Bayesian type I error and power requirements. If these requirements are met, then $N_{1}$ is deemed the Bayesian sample size for the study. If not, we evaluate the next sample size, $N_{2}$, and reassess its suitability for meeting the requirements. This process continues until we identify the Bayesian sample size meeting the requirements within the set $\mathcal{N}$. If it cannot be found within this set $\mathcal{N}$, it may be necessary to explore a broader range of candidate sizes, adjust the values of $\alpha$ and $\beta$ under regulatory consideration, modify the threshold $\lambda$, or consider other potential modifications such as changing the hyper-parameters of the fitting prior.

It is evident that Bayesian sample size determination is computationally intensive. It becomes even more intense when the complexity of the design increases. For instance, one needs to consider factors like the number and timing of interim analyses for Bayesian group sequential design, as well as the number of sub-groups and ratios in Bayesian platform design. Moreover, the computational complexity increases when the Bayesian test statistic requires MCMC sampling, as the convergence of the Markov chain should be diagnosed for each iteration within the Monte Carlo simulation. In such scenarios, the use of parallel computation techniques or modern sampling schemes can significantly reduce computation time \citep{ma2019sampling, hoffman2014no}.
\subsection{Calibration of Bayesian Trial Design to Assess Frequentist Operating Characteristics}\label{subsec:Calibration of Bayesian Trial Design}
Scientifically sound clinical trial planning and rigorous trial conduct are important, regardless of whether trial sponsors use a Bayesian or frequentist design. Maintaining some degree of objectivity in the interpretation of testing results is key to achieving scientific soundness. The central question here is how much we can trust a testing result based on a Bayesian hypothesis testing procedure, which is driven by the Bayesian type I error and power in the planning phase. More specifically, suppose that such a Bayesian test, where the threshold of the decision rule was chosen to meet the Bayesian type I error rate of less than 0.025 and power greater than 0.8, yielded the rejection of the null hypothesis, while a frequentist test did not upon completion of the study. Then, can we still use the result of the Bayesian test for registration purposes? Perhaps, this can be best addressed by calculating the frequentist type I error and power of the Bayesian test during the planning phase so that the Bayesian test can be compared with some corresponding frequentist test in an apple-to-apple comparison, or as close as possible.

In most regulatory submissions, Bayesian trial designs are `calibrated' to possess good frequentist properties. In this spirit, and in adherence to regulatory practice, regulatory agencies typically recommend that sponsors provide the frequentist type I and II error rates for the sponsor’s proposed Bayesian analysis plan \citep{grieve2016idle,Bayesian2010FDAGuidance}.

The simulation-based approach for Bayesian sample size determination \citep{wang2002simulation}, as illustrated in Subsection \ref{subsec:Simulation-based Bayesian Sample Size Determination}, is calibrated to measure the frequentist operating characteristics of a Bayesian trial design if the null sampling prior is specified by a Dirac measure with the point-mass at the boundary value of the null parameter space $\Theta_{0}$ (i.e., $\pi_{s0}(\theta)=\delta(\theta_{0})$ for some $\theta_{0} \in \Theta_{B} = \bar{\Theta}_{0} \cap \bar{\Theta}_{a}$ where $\delta(x)$ is the Direc-Delta function), and the alternative sampling prior is specified by a Dirac measure with the point-mass at the value inducing the minimally detectable treatment effect, representing the smallest effect size (i.e., $\pi_{s1}(\theta)=\delta(\theta_{a})$ for some $\theta_{a} \in \Theta_{a}^{*}\subset \Theta_{a}$).

In this calibration, the expected behavior of the Bayesian test statistics $T(\textbf{y}_N) = \mathbb{P}_f[\theta \in \Theta_a | \textbf{y}_N]$ can be represented as the frequentist type I error and power of the design as follow:
\begin{align}
\label{eq:freq_type_I_error_rate}
&\text{Type I error}: \beta_{\theta_{0}}^{(N)} = \mathbb{P}[T(\textbf{y}_{N}) > \lambda | \textbf{y}_{N} \sim f(\textbf{y}_{N} | \theta_{0})] = \mathbb{P}_{\theta_{0}}[T(\textbf{y}_{N}) > \lambda ],\\
\label{eq:freq_power}
&\text{Power}: \beta_{\theta_{a}}^{(N)} = \mathbb{P}[T(\textbf{y}_{N}) > \lambda | \textbf{y}_{N} \sim f(\textbf{y}_{N} | \theta_{a})] = \mathbb{P}_{\theta_{a}}[T(\textbf{y}_{N}) > \lambda ].
\end{align}

Throughout the paper, we interchangeably use the notations $\mathbb{P}_{\theta}[\cdot]$ and $\mathbb{P}[\cdot|\textbf{y}_{N} \sim f(\textbf{y}_{N}|\theta)]$. The former notation is simpler, yet it omits specifying which values are being treated as random and which are not; hence, the latter notation is sometimes more convenient for Bayesian computation.

With the aforementioned calibration, the prior specification problem of the Bayesian design essentially boils down to the choice of the fitting prior $\pi_{f}(\theta)$. This is because the selection of the null and alternative sampling prior is essentially determined by the formulation of the null and alternative hypotheses, aligning with the frequentist framework. In other words, the fitting prior provides the unique advantage of Bayesian design by incorporating prior information about the parameter $\theta$, which is then updated by Bayes’ theorem, leading to the posterior distribution. The choice of the fitting prior will be discussed in Section \ref{sec:Specification of Prior Distributions}. In what follows, to avoid notation clutter, when we omit the subscript `$f$' in the notation of the fitting prior $\pi_{f}(\theta)$. 

\subsection{Example - Standard Single-stage Design Based on Beta-Binomial Model}\label{subsec:Example - Standard Single-stage Design Based on Beta-Binomial Model}
Suppose a medical device company aims to evaluate the primary safety endpoint of a new device in a pivotal trial. The safety endpoint is the primary adverse event rate through 30 days after a surgical procedure involving the device. The sponsor plans to conduct a single-arm study design in which patient data is accumulated throughout the trial. Only once the trial is complete, the data will be unblinded, and the pre-planned statistical analyses will be executed. Suppose that the null and alternative hypotheses are: \(\mathcal{H}_{0}: \theta \geq \theta_{0}\) versus \(\mathcal{H}_{a}: \theta < \theta_{0}\). Here, \(\theta_{0}\) represents the performance goal of the new device, a numerical value (point estimate) that is considered sufficient by a regulator for use as a comparison for the safety endpoint. It is recommended that the performance goal not originate from a particular sponsor or regulator. It is often helpful if it is recommended by a scientific or medical society. \citep{DesignConsiderations2013FDAGuidance}.

A fundamental regulatory question is ``when a device passes a safety performance goal, does that provide evidence that the device is safe?". To answer this question, the sponsor sets a performance goal by \(\theta_{0} = 0.12\), and anticipates that the safety rate of the new device is \(\theta_{a} = 0.05\). The objective of the study is, therefore, to detect a minimum treatment effect of \(7\% = 12\% - 5\%\) in reducing the adverse event rate of patients treated with the new medical device compared to the performance goal. The sponsor targeted to achieve a statistical power of $1-\beta = 0.8$ with the one-sided level \(\alpha = 0.025\) test of a proposed design. The trial is successful if the null hypothesis \(\mathcal{H}_{0}\) is rejected after observing the outcomes from \(N\) patients upon completion of the study.

The following Bayesian design is considered:
\begin{itemize}
\baselineskip=12pt
\item One-sided significance level: $\alpha = 0.025$,
\item Power: $1 - \beta= 0.8$,
\item Null sampling prior: $\pi_{s0}(\theta) = \delta(\theta_{0})$, where $\theta_{0}= 0.12$,
\item Alternative sampling prior: $\pi_{s1}(\theta) = \delta(\theta_{a})$, where $\theta_{a}= 0.05$,
\item Prior: $\theta \sim \pi(\theta) = \mathcal{B}eta(\theta|a,b)$,
\item Hyper-parameters: $a>0$ and $b>0$,
\item Likelihood: $y_{i}\sim f(y|\theta) = \mathcal{B}ernoulli(y|\theta),\, (i=1,\cdots,N)$,
\item Decision rule: Reject null hypothesis if $T(\textbf{y}_{N}) = \mathbb{P}[\theta < \theta_{0} | \textbf{y}_{N}] > 0.975$.
\end{itemize}

Under the setting, (frequentist) type I error and power of the Bayesian design can be expressed as:
\begin{align*}
\beta_{\theta_{0}}^{(N)} &= \mathbb{P}_{\theta_{0}}[ \mathbb{P}[\theta < \theta_{0} | \textbf{y}_{N}] > 0.975 ] = \int \textbf{1}( \mathbb{P}[\theta < \theta_{0} | \textbf{y}_{N}] > 0.975 ) \cdot \prod_{i=1}^{N} \theta_{0}^{y_{i}} (1-\theta_{0})^{1-y_{i}}d\textbf{y}_{N},\\
\beta_{\theta_{a}}^{(N)} &= \mathbb{P}_{\theta_{0}}[ \mathbb{P}[\theta < \theta_{0} | \textbf{y}_{N}] > 0.975 ] = \int \textbf{1}( \mathbb{P}[\theta < \theta_{0} | \textbf{y}_{N}] > 0.975 ) \cdot \prod_{i=1}^{N} \theta_{a}^{y_{i}} (1-\theta_{a})^{1-y_{i}}d\textbf{y}_{N}.
\end{align*}
The Bayesian sample size satisfying the type I \& II error requirements are then $$N = \text{max}(\text{min}\{N : \beta_{\theta_{0}}^{(N)} \leq 0.025\},
 \text{min}\{N : \beta_{\theta_{a}}^{(N)} \geq 0.8 \} ).$$

Due the conjugate relationship between the binomial distribution and beta prior, the posterior distribution is the beta distribution, $\pi(\theta|\textbf{y}_{N})= \mathcal{B}eta(x + a, N - x + b)$ such that $x = \sum_{i=1}^{N}y_{i}$. Therefore, the Bayesian test statistics $T(\textbf{y}_{N})=\mathbb{P}[\theta < \theta_{0} | \textbf{y}_{N}]$ can be represented as a closed-form in this case.

We consider $N= 100, 150,$ and 200 as the possible sizes for the Bayesian trial. We evaluate three prior options: (1) a non-informative prior with $a = b = 1$ (prior mean is 50\%), (2) an optimistic prior with $a = 0.8$ and $b = 16$ (prior mean is 4.76\%), and (3) a pessimistic prior with $a = 3.5$ and $b = 20$ (prior mean is 14.89\%). An optimistic prior assigns a probability mass that is favorable for rejecting the null hypothesis before observing any new outcomes, while a pessimistic prior assigns a probability mass that is favorable for accepting the null hypothesis before observing any new outcomes. As a reference, we consider a frequentist design in which the decision criterion is determined by the p-value associated with the z-test statistic, $Z = (x/N - \theta_{0})/\sqrt{\theta_{0}(1 - \theta_{0})/N}$, being less than the one-sided significance level of $\alpha=0.025$ to reject the null hypothesis.

\begin{table}[h]
\caption{Frequentist operating characteristics of Bayesian designs with different prior options.}
		\newcolumntype{C}{>{\centering\arraybackslash}X}
\begin{footnotesize}
\begin{tabular}{ccccccccc}
\toprule
\multirow{2}{*}{Sample Size ($N$)} & \multicolumn{2}{c}{\begin{tabular}[c]{@{}c@{}}Bayesian Design\\ (Non-informative prior)\end{tabular}} & \multicolumn{2}{c}{\begin{tabular}[c]{@{}c@{}}Bayesian Design\\ (Optimistic prior)\end{tabular}} & \multicolumn{2}{c}{\begin{tabular}[c]{@{}c@{}}Bayesian Design\\ (Pessimistic prior)\end{tabular}} & \multicolumn{2}{c}{\begin{tabular}[c]{@{}c@{}}Frequentist Design\\ (Z-test statistics)\end{tabular}} \\ \cline{2-9} 
                             & Type I Error                                         & Power                                          & Type I Error                                       & Power                                       & Type I Error                                       & Power                                        & Type I Error                                         & Power                                         \\ \midrule
100                        & 0.0148                                               & 0.6181                                        & 0.0755                                             & 0.8767                                      & 0.0148                                             & 0.6181                                       & 0.0155                                               & 0.6214                                        \\
150                        & \textbf{0.0231                                             }   & \textbf{0.8690}                                         & 0.0448 & 0.9268 & 0.0114                                             & 0.7838                                       & \textbf{0.0242}                                               & \textbf{0.8690}                                        \\
200                        & \textbf{0.0164                                           } & \textbf{0.9184}                                         & 0.0467	 & 0.9767 & 0.0164                                            & 0.9184                                       & \textbf{0.0158}                                               & \textbf{0.9231}                                        \\ 
\bottomrule
\end{tabular}
\\
\label{tab:Beta-binomial_single_arm_design}
\baselineskip=12pt
Note: Bayesian designs are based on the beta-binomial models with prior options: (1) a non-informative prior with $a = b = 1$, (2) an optimistic prior with $a = 0.8$ and $b = 16$, and (3) a pessimistic prior with $a = 3.5$ and $b = 20$.
\end{footnotesize}
\end{table}

Table \ref{tab:Beta-binomial_single_arm_design} shows the results of the power analysis obtained by simulation. Designs satisfying the requirement of type I error $\leq$ 2.5\% and power $\geq$ 80\%, are highlighted in bold in the table. The results indicate that the operating characteristics of the Bayesian design based on a non-informative prior are very similar to those obtained using the frequentist design. This similarity is typically expected because a non-informative prior has minimal impact on the posterior distribution, allowing the data to play a significant role in determining the results.

The results show that the Bayesian design based on an optimistic prior tends to increase power at the expense of inflating the type I error. Technically, the inflation is expected because, by definition, the type I error is evaluated by assuming the true treatment effect is null (i.e. $\theta = \theta_{0}$), then it is calculated under a scenario where the prior is in conflict with the null treatment effect, resulting in the inflation of the type I error. In contrast, the Bayesian design based on a pessimistic prior tends to decrease the type I error at the cost of deflating the power. The deflation is expected because, by definition, the power is evaluated by assuming the true treatment effect is alternative (i.e. $\theta = \theta_{a}$), then it is calculated under a scenario where the prior is in conflict with the alternative treatment effect, resulting in the deflation of the power.

Considering the trade-off between power and type I error, which is primarily influenced by the prior specification, thorough pre-planning is essential for selecting the most suitable Bayesian design on a case-by-case basis for regulatory submission. Particularly, when historical data is incorporated into the hyper-parameter of the prior as an optimistic prior, there may be inflation of the type I error rate, even after appropriately discounting the historical data \citep{burger2021use}. In such cases, it may be appropriate to relax the type I error control to a less stringent level compared to situations where no prior information is used. This is because the power gains from using external prior information in clinical trials are typically not achievable when strict type I error control is required \citep{best2023beyond,kopp2020power}. Refer to Section 2.4.3 in \citep{lesaffre2020bayesian} for relevant discussion. The extent to which type I error control can be relaxed is a case-by-case decision for regulators, depending on various factors, primarily the confidence in the prior information \citep{Bayesian2010FDAGuidance}. We discuss this in more detail by taking the Bayesian borrowing design based on a power prior \citep{ibrahim2000power} as an example in Section \ref{sec:Historical Data Information Borrowing}.

\subsection{Numerical Approximation of Power Function}\label{subsec:Numerical Approximation of Power Function}
In this subsection, we illustrate a numerical method to approximate the power function of a Bayesian hypothesis testing procedure. The power function of a test procedure is the probability of rejecting the null hypothesis, with the true parameter value as the input. The power function plays a crucial role in assessing the ability of a statistical test to detect a true effect or relationship between the design parameters. Visualizing the power function over the parameter space, as provided by many statistical software (SAS, PASS, etc), is helpful for trial sizing because it displays the full spectrum of the behavior of the testing procedure. Understanding such behaviors is crucial for regulatory submission, as regulators often recommend simulating several likely scenarios and providing the expected sample size and estimated type I error for each case.

Consider the null and alternative hypotheses, $\mathcal{H}_{0}: \theta \in \Theta_{0}$ versus $\mathcal{H}_{a}: \theta \in \Theta_{a}$, where $\Theta = \Theta_{0} \cup \Theta_{a}$, and $\Theta_{0}$ and $\Theta_{a}$ are disjoint. Let outcomes $y_{i}$ ($i=1,\cdots,N$) be identically and independently distributed according to a density $f(y|\theta)$. Given a Bayesian test statistics $T(\textbf{y}_{N})$, suppose that a higher value of $T(\textbf{y}_{N})$ raises more doubt about the null hypothesis being true. We reject the null hypothesis if $T(\textbf{y}_{N})>\lambda$, where $\lambda$ is a pre-specified threshold. Then, the power function $\psi: \Theta \rightarrow [0,1]$ is defined as follows:
\begin{align}
\nonumber
\psi(\theta) &= \mathbb{P}_{\theta}[T(\textbf{y}_{N}) > \lambda]\\ 
\nonumber
&=
\mathbb{P}[T(\textbf{y}_{N}) > \lambda |\textbf{y}_{N} \sim f(y|\theta) ]
\\
\label{eq:power_function}
&=\int \textbf{1}\{T(\textbf{y}_{N}) > \lambda\} \prod_{i=1}^{n} f(y_{i}|\theta) d\textbf{y}_{N}.
\end{align}

Eventually, one needs to calculate $\psi(\theta)$ over the entire parameter space $\Theta$ to explore the behavior of the testing procedure. However, the value of $\psi(\theta)$ is often not expressed as a closed-form formula, mainly due to two reasons: no explicit formula for the outside integral $\mathbb{P}_{\theta}[\cdot]$ or the Bayesian test statistics $T(\textbf{y}_{N})$. Thus, it is often usual that the value of $\psi(\theta)$ is approximated through a nested simulation strategy. See Algorithm \ref{alg:nested simulation} that prints out $\widetilde{\psi}(\theta)$ to approximate $\psi(\theta)$ at the parameter $\theta \in \Theta$. The idea is that the outside integral is approximated by a Monte-Carlo simulation (with $R$ number of replicated studies), and the test statistics is approximated by Monte-Carlo or Markov Chain Monte-Carlo simulation (with $S$ number of posterior samples) when the test statistics are not expressed in closed form. It is important to note that this approximation is exact in the sense that if $R$ and $S$ go to infinity, then $\widetilde{\psi}(\theta)$ converges to the truth $\psi(\theta)$. This contrasts with the formulation of the power functions of many frequentist tests, which are derived based on some large sample theory \citep{hall1990large}, to induce a closed-form formula.


\begin{algorithm}[h!]
\baselineskip=5pt
\caption{A nested simulation to approximate the power function $\psi(\theta)$ }\label{alg:nested simulation}
\SetAlgoLined
\textbf{Goal : } Approximating the power function $\psi(\theta)$ (\ref{eq:power_function}) evaluated at $\theta\in \Theta$.
\\
\textbf{Input : } A prior $\pi(\theta)$, data generating distribution $f(y|\theta)$, threshold value $\lambda\in [0,1]$, Bayesian test statistics $T(\textbf{y}_{N})$, true data generating parameter $\theta$, number of repetitions of trials $R$, and number of posterior samples $S$.\\
\textbf{Output : } An approximated value $\widetilde{\psi}(\theta)$.\\
\begin{itemize}
\item[] \textbf{for}  \textbf{(} $r = 1,\cdots,R$  \textbf{)} $\textbf{\{}$ 
\\
\indent $\quad \quad $ Generate the synthetic responses of $N$ patients assuming that $\theta$ is true \\
\indent $\quad \quad $ \indent $\textbf{y}_{N}^{(r)} = (y_{1}^{(r)},\cdots,y_{N}^{(r)})\sim f(y|\theta)$
\\
\indent $\quad \quad \quad$  \textbf{for}  \textbf{(} $s = 1,\cdots,S$  \textbf{)} $\textbf{\{}$ 
\\
\indent $\quad \quad \quad \quad \quad $ Sample from posterior distribution given $\textbf{y}_{N}^{(r)}$
\\
\indent $\quad \quad \quad \quad \quad $  $\theta^{(s)} \sim \pi(\theta|\textbf{y}_{N}^{(r)})$ 
\\
\indent $\quad \quad \quad \quad \quad \quad$  
\\
\indent  $\quad \quad \quad\quad\quad$ $\textbf{\}}$ 
\\
\indent  $\quad \quad $ Approximate Bayesian test statistics $T(\textbf{y}_{N})$ by $\widehat{T}(\textbf{y}_{N}) = g(\theta^{(1)},\cdots,\theta^{(S)})$
\\
\indent  $\quad $  $\textbf{\}}$ 
\\
Approximate power function $\psi(\theta)$ by 
$\widetilde{\psi}(\theta) =
\frac{1}{R}
\sum_{r=1}^{R}
1\left(
\widehat{T}\{\textbf{y}_{N}\}
 > \lambda
\right)$
\end{itemize}
\end{algorithm}

\subsection{Example - Bayesian Hypothesis Testing for the Degree of Freedom of Student t-distribution}\label{subsec:Example - Bayesian Hypothesis Testing for the Degree of Freedom of Student t-distribution}
We apply Algorithm \ref{alg:nested simulation} to an instructive example where the objective is to test concerning the number of degrees of freedom of Student t-distribution. Suppose that the outcomes $y_{i}$ ($i=1,\cdots,N$) are identically and independently distributed according to the t-distribution with the degrees of freedom $\theta$: $$ y_{i}\sim f(y|\theta) = \frac{\Gamma((\theta+1)/2)}{\sqrt{\theta \pi} \Gamma(\theta/2)} \cdot \bigg(1 + \frac{y^{2}}{\theta} \bigg)^{- \frac{\theta +1}{2}}.$$ It is widely known that frequentist estimation of the parameter $\theta$ is difficult due to the leptokurtic nature of the Student t-distribution, but Bayesian estimation may be more robust than frequentist estimation \citep{villa2018objective}. We consider the null and alternative hypotheses, $\mathcal{H}_{0}: \theta \geq 5$ versus $\mathcal{H}_{a}: \theta < 5$. To execute the Bayesian hypothesis test, we assume a log-normal prior for the degrees of freedom, given as $\pi(\theta)= \log\mathcal{N}(\theta|1,1) = \{1/(\theta\sqrt{2\pi})\} \exp [-\{\log (\theta) -1\}^{2}/2 ], \, \theta \in (0,\infty)$, studied by \citep{lee2022use}. We use the Bayesian test statistic $T(\textbf{y}_{N}) = \mathbb{P}[\theta < 5 | \textbf{y}_{N}]$, and reject the null if $T(\textbf{y}_{N})>\lambda$, where $\lambda = 0.985$. The threshold $\lambda$ was chosen to protect the type I error rate less than 2.5\%. We repeat the trials $R=50,000$ times and use the elliptical slice sampler \citep{murray2010elliptical} to sample from the posterior distribution $\theta^{(s)} \sim \pi(\theta|\textbf{y}_{N})$, $(s=1,\cdots,S)$ where $S=3,000$, in order to approximate the value of the power function. The MCMC technique is necessary because the posterior distribution does not have a closed-form expression. 

\begin{figure}[h!]
\centering
\includegraphics[width=\textwidth]{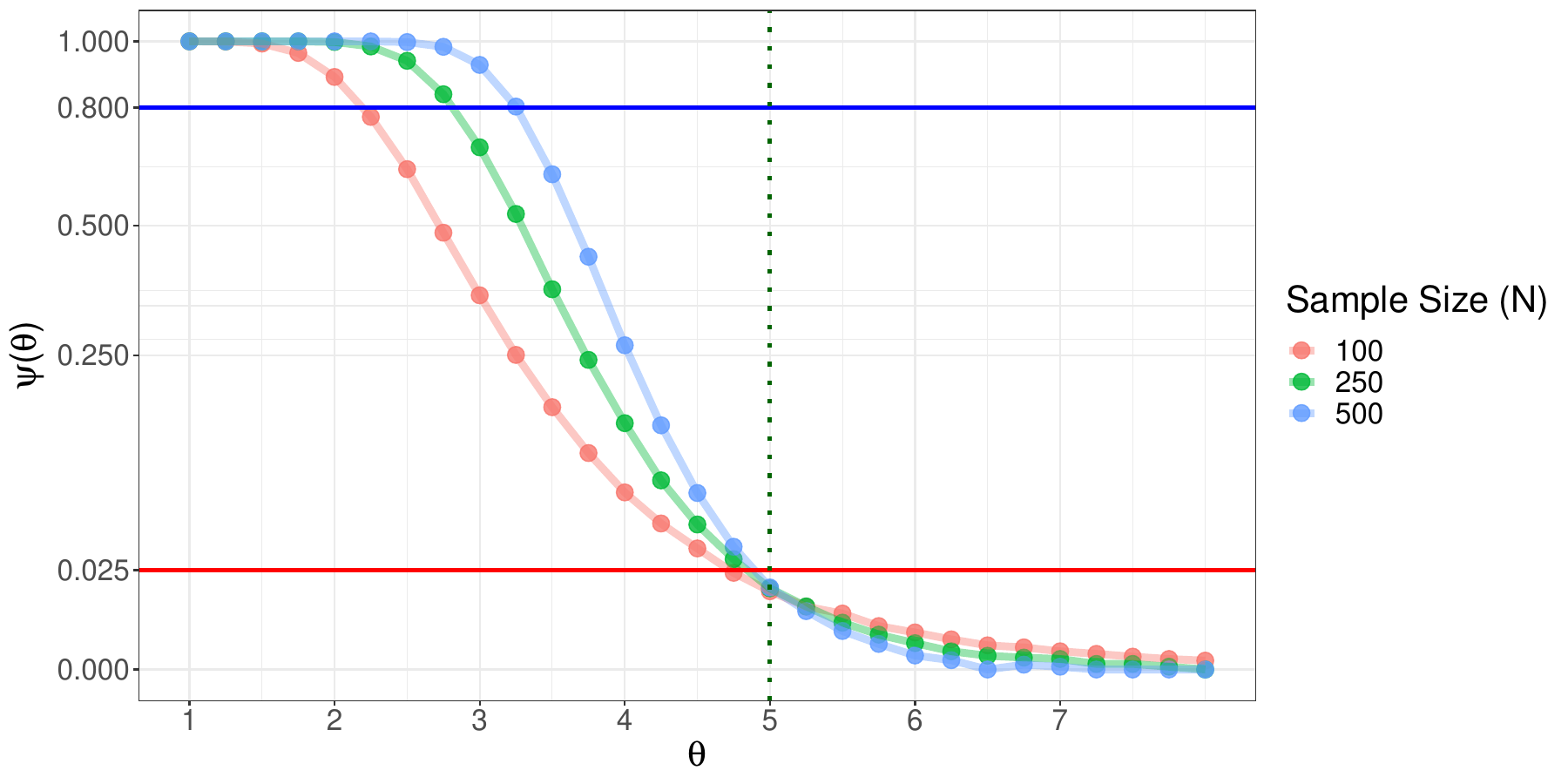}
\caption{\baselineskip=12pt Power functions of Bayesian hypothesis testing for null and alternative hypotheses, $\mathcal{H}_{0}: \theta \geq 5$ versus $\mathcal{H}_{a}: \theta < 5$ with sample size of $N=100, 250, $ and $500$. Here, $\theta $ denotes the number of degrees of freedom of Student t-distribution.}
\label{fig:power_ft_student_t}
\end{figure}

Figure \ref{fig:power_ft_student_t} displays the power functions for a sample size of $N=100, 250,$ and $500$. The type I error rates are $0.0155$ ($N=100$), $0.0167$ ($N=250$), and $0.0171$ ($N=500$), all kept below 2.5\%. This implies that all test procedures maintain a one-sided level-$\alpha$ test, where $\alpha=2.5\%$. It is observed that a larger sample size leads to a more powerful test in detecting the minimum effect size, which is typically shown in the frequetist testing procedures. 

\section{Specification of Prior Distributions}\label{sec:Specification of Prior Distributions}
\subsection{Classes of Prior Distributions}\label{subsec:Classes of Prior Distributions}
The prior distributions for regulatory submissions can be broadly classified into non-informative priors and informative priors. A non-informative prior is a prior distribution with no preference for any specific parameter value. A Bayesian design based on a non-informative prior leads to objective statistical inference, resembling frequentist inference, and is therefore the least controversial. It is important to note that choosing a non-informative prior distribution can sometimes be challenging, either because there may be more than one way to parameterize the problem or because there is no clear mathematical justification for defining non-informativeness. \citep{kass1996selection} reviews the relevant literature but emphasizes the continuing difficulties in defining what is meant by `non-informative' and the lack of agreed reference priors in all but simple situations.

For example, in the case of a beta-binomial model (as illustrated in Subsection \ref{subsec:Calibration of Bayesian Trial Design}), choices such as $\mathcal{B}eta(1,1)$, $\mathcal{B}eta(0.5,0.5)$, or $\mathcal{B}eta(0.001,0.001)$ could all be used as non-informative priors. See Subsection 5.5.1 from \citep{spiegelhalter2004bayesian} for a relevant discussion. In Bayesian hierarchical models, the mathematical meaning of a non-informative prior distribution is not obvious due to the complexity of the model. In those cases, we typically set the relevant hyper-parameters to diffuse the prior evenly over the parameter space and minimize the prior information as much as possible, leading to a nearly non-informative prior.

On the other hand, an informative prior is a prior distribution that expresses a preference for a particular parameter value, enabling the incorporation of prior information. Informative priors can be further categorized into two types: prior distributions based on empirical evidence from previous trials and prior distributions based on personal opinions, often obtained through expert elicitation. The former class of informative priors is less controversial when the current and previous trials are similar to each other. Possible sources of prior information include: clinical trials conducted overseas, patient registries, clinical data on very similar products, and pilot studies. Recently, there has been breakthrough development of informative prior distribution that enables incorporating the information from previous trials, and eventually reducing sample size of a new trial, while providing appropriate mechanism of discounting \citep{ibrahim2015power,ibrahim2003optimality,
thall2003hierarchical,lee2022bayesian}. We provide details on the formulation of an informative prior and relevant regulatory considerations in Section \ref{sec:Historical Data Information Borrowing}. Typically, informative prior distribution based on personal opinions is not recommended for Bayesian submissions due to subjectivity and controversy \citep{irony2001choosing}. 

Incorporating prior information formally into the statistical analysis is a unique feature of the Bayesian approach but is also often criticized by non-Bayesians. To mitigate any conflict and skepticism regarding prior information, it is crucial that sponsors and regulators meet early in the process to discuss and agree upon the prior information to be used for Bayesian clinical trials. 
\subsection{Prior probability of the study claim}\label{subsec:Prior probability of the study claim}
The prior predictive distribution plays a key role in pre-planning a Bayesian trial to measure the prior probability of the study claim—the probability of the study claim before observing any new data. Regulators recommend that this probability should not be excessively high, and what constitutes `too high' is a case-by-case decision \citep{Bayesian2010FDAGuidance}. Measuring this probability is typically recommended when an informative prior distribution is used for the Bayesian submission. Regulatory agencies make this recommendation to ensure that prior information does not overwhelm the data of a new trial, potentially creating a situation where unfavorable results from the proposed study get masked by a favorable prior distribution. In an evaluation of the prior probability of the claim, regulators will balance the informativeness of the prior against the efficiency gain from using prior information, as opposed to using noninformative priors.

To calculate the prior probability of the study claim, we simulate multiple hypothetical trial data using the prior predictive distribution (\ref{eq:prior_pred_dist}) by setting the sampling prior as the fitting prior, and then calculate the probability of rejecting the null hypothesis based on the simulated data. We illustrate the procedure for calculating this probability using the beta-binomial model illustrated in Subsection \ref{subsec:Calibration of Bayesian Trial Design} as an example. First, we generate the data $(\textbf{y}_{N})^{(r)} \sim f(\textbf{y}_{N}) = \int f(\textbf{y}_{N}|\theta) \pi(\theta) d\theta$ ($r=1,\cdots,R$), where $R$ represents the number of simulations. Here, $f$ is the Bernoulli likelihood, and $\pi$ is the beta prior with hyper-parameters $a$ and $b$. In this particular example, $a$ and $b$ represent the number of hypothetical patients showing adverse events and not showing adverse events \emph{a priori}, hence $a+b$ is the prior effective sample size. The number of patients showing adverse events out of $N$ patients, $X^{(r)} = \sum_{i=1}^{N}y_{i}^{(r)}$, is distributed according to a beta-binomial distribution \citep{griffiths1973maximum}, denoted as $X^{(r)} \sim \mathcal{B}eta$-$\mathcal{B}inom(N,a,b)$. One can use a built-in function $\mathsf{rbetabinom.ab(\cdot)}$ within the $\mathsf{R}$ package $\mathsf{VGAM}$ to generate the $r$-th outcome $X^{(r)}$. Second, we compute the posterior probability and make a decision whether to reject the null or not, i.e., $d(r)=\textbf{1}\{\mathbb{P}[\theta < \theta_{0} | \textbf{y}_{N}^{(r)}] > 0.975 \} = 1$ if $\mathcal{H}_{0}$ is rejected and $0$ otherwise. Finally, the value of $\sum_{r=1}^{R}d(r)/R$ is the prior probability of the study claim based on the prior choice of $\theta \sim \pi(\theta) = \mathcal{B}eta(\theta|a,b)$.

We consider four prior options where the hyperparameters have been set to induce progressively stronger prior information to reject the null \emph{a priori}. Table \ref{tab:Prior probability of the study claim based on beta-binomial model} shows the results of the calculations of this probability. For the non-informative prior, the prior probability of the study claim is only 5.8\%, implying that the outcome from a new trial will most likely dominate the final decision. However, the third and fourth options provide probabilities greater than 50\%, indicating overly strong prior information; hence, appropriate discounting on the prior effective sample size is recommended. 

\begin{table}[h]
\centering
\caption{Prior probability of the study claim based on beta-binomial model.}
		\newcolumntype{C}{>{\centering\arraybackslash}X}
		\begin{scriptsize}
\begin{tabular}{ccccc}
\toprule
Prior Distribution     & \begin{tabular}[c]{@{}c@{}}Number of hypothetical patients \\ showing adverse events\end{tabular} & \begin{tabular}[c]{@{}c@{}}Number of hypothetical patients \\ not showing adverse events\end{tabular} & \begin{tabular}[c]{@{}c@{}}Prior Mean \\ (Standard Deviation)\end{tabular} & \begin{tabular}[c]{@{}c@{}}Prior Probability \\ of a Study Claim\end{tabular} \\
\midrule
$\mathcal{B}eta(1,1)$  & 1                                                                                              & 1                                                                                                     & 50\% (5.8\%)                                                                 & 5.8\%                                                                          \\
$\mathcal{B}eta(1,9)$  & 1                                                                                              & 9                                                                                                     & 10\% (9\%)                                                                   & 47.1\%                                                                         \\
$\mathcal{B}eta(1,19)$ & 1                                                                                              & 19                                                                                                    & 5\% (4.9\%)                                                                  & 77.3\%                                                                         \\
$\mathcal{B}eta(1,49)$ & 1                                                                                              & 49                                                                                                    & 2\% (2\%)                                                                    & 99.1\% \\          
\bottomrule                                                             
\end{tabular}
		\end{scriptsize}
\label{tab:Prior probability of the study claim based on beta-binomial model}
\end{table}

\section{Decision Rule - Posterior Probability Approach}\label{sec:Bayesian Decision Rule - Posterior Probability Approach}

\subsection{Posterior Probability Approach}\label{subsec:Posterior Probability Approach}
The central motivation for utilizing the posterior probability approach in decision-making is to quantify the evidence to address the question, ``Does the current data provide convincing evidence in favor of the alternative hypothesis?" The key quantity here is the posterior probability of the alternative hypothesis being true based on the data observed up to the point of analysis. This Bayesian tail probability can be used as the test statistic in a single-stage Bayesian design upon completion of the study, similar to the role of the p-value in a single-stage frequentist design \citep{lesaffre2020bayesian}. Furthermore, one can measure it in both interim and final analyses within the context of Bayesian group sequential designs \citep{stallard2020comparison, gsponer2014practical}, akin to a z-score in a frequentist group sequential design \citep{fleming1984designs, jennison1999group}.

It is important to note that if the posterior probability approach is used in decision-making at the interim analysis, it does not involve predicting outcomes of the future remaining patients. This distinguishes it from the predictive probability approach, where the remaining time and statistical information to be gathered play a crucial role in decision-making at the interim analysis (as discussed in Section \ref{sec:Bayesian Decision Rule - Predictive Probability Approach}). Consequently, the posterior probability approach is considered conservative, as it prohibits imputation for incomplete data or partial outcomes. For this reason, the posterior probability approach is standardly employed in interim analyses to declare early success or in the final analysis to declare the trial's success to support marketing approval of medical devices or drugs in the regulatory submissions \citep{polack2020safety,bohm2021re}.

Suppose that $\textbf{y}$ denotes an analysis dataset, and $\theta$ is the parameter of main interest. A sponsor wants to test $\mathcal{H}_{0}: \theta \in \Theta_{0}$ versus $\mathcal{H}_{a}: \theta \in \Theta_{a}$, where $\Theta = \Theta_{0} \cup \Theta_{a}$, and $\Theta_{0}$ and $\Theta_{a}$ are disjoint. Bayesian test statistics following the posterior probability approach can be represented as a functional $\mathcal{F}\{\cdot\}: \mathcal{Q}_{\theta|\textbf{y}} \rightarrow [0,1]$, such that:

\begin{align}
\label{eq:Bayesian_test_statistics_Post_Prob_Approach}
\mathcal{F}\{\pi(\theta|\textbf{y})\} = T(\textbf{y}) & = \mathbb{P}[\theta \in \Theta_{a} | \textbf{y}] = \int \textbf{1}(\theta \in \Theta_{a}) \cdot \pi(\theta|\textbf{y}) d\theta,
\end{align}

where $\mathcal{Q}_{\theta|\textbf{y}}$ represents the collection of posterior distributions. Finally, to induce a dichotomous decision, we need to pre-specify the threshold $\lambda\in [0,1]$. By introducing an indicator function $\varphi$ (referred as a `critical function' in \citep{lehmann1986testing}), the testing result is determined as follow:
\begin{align*}
\varphi(\textbf{y}) = \begin{cases}
  1  & \text{ if } \mathcal{F}\{\pi(\theta|\textbf{y})\} = \mathbb{P}[\theta \in \Theta_{a} | \textbf{y}] > \lambda\\
  0   & \text{ if } \mathcal{F}\{\pi(\theta|\textbf{y})\} = \mathbb{P}[\theta \in \Theta_{a} | \textbf{y}] \leq \lambda,
\end{cases}
\end{align*}
where $1$ and $0$ indicate the rejection and acceptance of the null hypothesis, respectively. 

In the interim analysis, rejecting the null can be interpreted as claiming the early success of the trial, and in the final analysis, rejecting the null can be interpreted as claiming the final success of the trial. Figure \ref{fig:Post_Prob_Approach} displays a pictorial description of the decision procedure.

\begin{figure}[h!]
\centering
\includegraphics[scale=0.35]{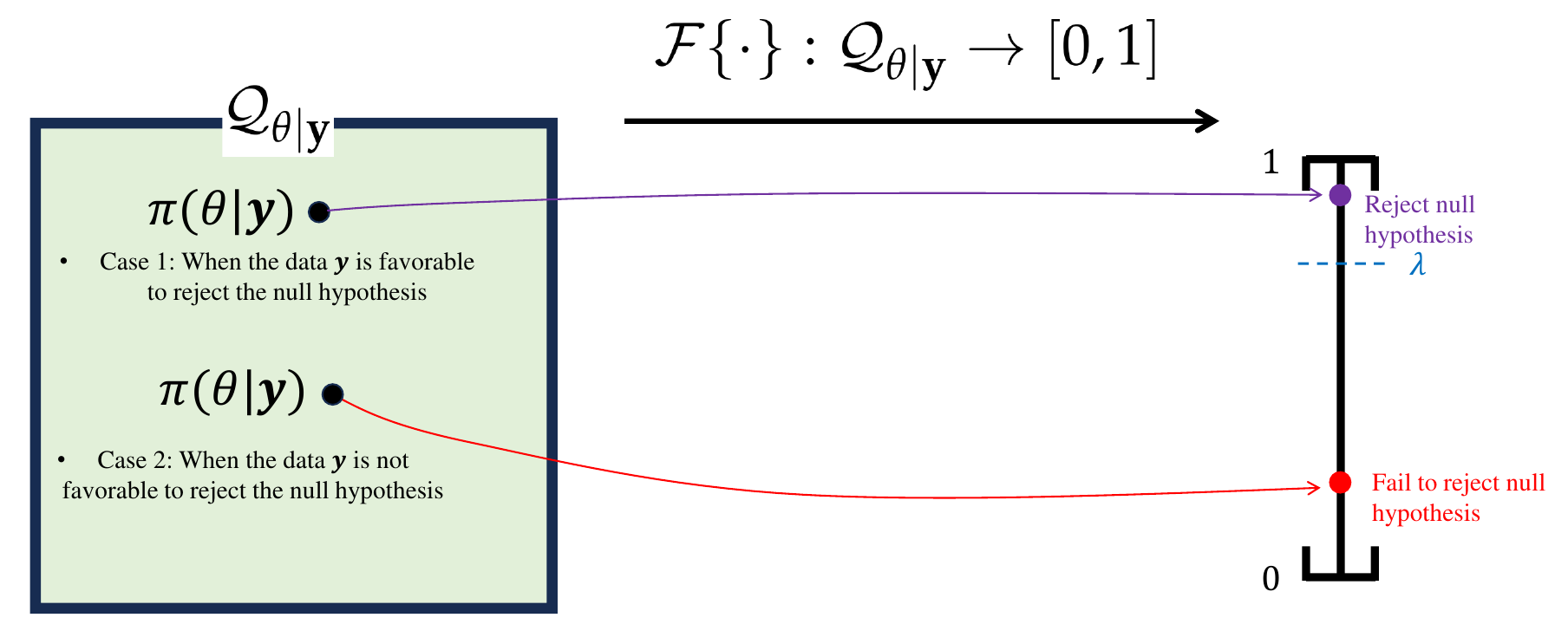}
\caption{\baselineskip=12pt Pictorial illustration of the decision rule based on the posterior probability approach: If the data $\textbf{y}$ were generated from the alternative (or null) density $f(\textbf{y}|\theta)$ where $\theta \in \Theta_{a}$ (or $\theta \in \Theta_{0}$), then the posterior distribution would be more concentrated on the alternative space $\Theta_{a}$ (or null parameter $\Theta_{0}$), resulting in a higher (or lower) value of the test statistic $\mathcal{F}\{ \pi(\theta|\textbf{y})\} = \mathbb{P}[\theta \in \Theta_{a} | \textbf{y}]$. The pre-specified threshold $\lambda$ is used to make the dichotomous decision based on the test statistic.}
\label{fig:Post_Prob_Approach}
\end{figure}

The formulation of Bayesian test statistics is universal regardless of the hypothesis being tested (e.g., mean comparison, proportion comparison, association), and it does not rely on asymptotic theory. The derivation procedure for Bayesian test statistics based on the posterior probability approach is intuitive, considering the backward process of the Bayesian theorem. A higher value of $T(\textbf{y})=\mathbb{P}[\theta \in \Theta_{a} | \textbf{y}]$ implies that more mass has been concentrated on the alternative parameter space $\Theta_{a}$ \emph{a posteriori}. Consequently, there is a higher probability that the data were originally generated from the density indexed with parameters belonging to $\Theta_{a}$, that is, $\textbf{y}\sim f(\textbf{y}|\theta)$, $\theta \in \Theta_{a}$. The prior distribution in this backward process acts as a moderator by appropriately allocating even more or less mass on the parameter space $\Theta$ before seeing any data $\textbf{y}$. If there is no prior information, the prior distribution plays a minimal role in this process.

This contrasts with the derivation procedure for frequentist test statistics, which involves formulating a point estimator such as sufficient statistics from the sample data to make a decision about a specific hypothesis. The derivation may vary depending on the type of test (e.g., t-test, chi-squared test, z-test) and the hypothesis being tested. Furthermore, asymptotic theory is often used if the test statistics based on exact calculation are difficult to obtain \citep{chow2017sample}.

For a single-stage design with the targeted one-sided significance level of $\alpha$, the threshold $\lambda$ is normally set to $1-\alpha$, provided that the test is a one-sided test and the prior distribution is a non-informative prior. This setting is frequently chosen, particularly when there is no past historical data to be incorporated into the prior; see the example of the beta-binomial model in Subsection \ref{subsec:Calibration of Bayesian Trial Design}. If an informative prior is used, this practice (that is, $\lambda = 1-\alpha$) should be carefully used  because the type I error rate can be inflated or deflated based on the direction of the informativeness (see Table \ref{tab:Beta-binomial_single_arm_design}).

\subsection{Asymptotic Property of Posterior Probability Approach}\label{subsec:Asymptotic Property of Posterior Probability Approach}
Bernstein-Von Mises theorem \citep{johnstone2010high,walker1969asymptotic}, also called Bayesian central limit theorem, states that if the sample size $N$ is sufficiently large, the influence of the prior $\pi(\theta)$ diminishes, and the posterior distribution $\pi(\theta|\textbf{y}_{N})$ closely resembles the likelihood $f(\textbf{y}_{N}|\theta)$ under suitable regularity conditions (for e.g., conditions stated in \citep{walker1969asymptotic} or Section 4.1.2 of \citep{ghosh2006introduction}). Consequently, it simplifies the complex posterior distribution into a more manageable normal distribution, independent of the form of prior, as long as the prior distribution is continuous and positive on the parameter space. 

By using Bernstein-Von Mises theorem, we can show that if the sample size $N$ is sufficiently large, the posterior probability approach asymptotically behaves similarly to the frequentist testing procedure based on the p-value approach \citep{fisher1936design} under the regularity conditions. For the ease of exposition, we consider a one-sided testing problem. In this specific case, we further establish an asymptotic equation between the Bayesian tail probability (\ref{eq:Bayesian_test_statistics_Post_Prob_Approach}) and p-value. 

\begin{theorem}\label{thm:Bayes_Freq_same}
Let a random sample of size $N$, \(y_i,\ (i=1,\ldots,N)\), be independently and identically taken from a distribution \(f(y|\theta)\) depending on the real parameter \(\theta \in \Theta \subset \mathbb{R}\). Consider a one-sided testing problem \(\mathcal{H}_{0}: \theta \leq \theta_{0}\) versus \(\mathcal{H}_{A}: \theta > \theta_{0},\) where \(\theta_{0}\) denotes the performance goal. Consider testing procedures with two paradigms:
\begin{align*}
\text{Frequentist testing procedure} & : T_{1}(\mathbf{y}_{N}) > \lambda_{1} \Longleftrightarrow \text{Reject } \mathcal{H}_{0}; \\
\text{Bayesian testing procedure} & : T_{2}(\mathbf{y}_{N}) > \lambda_{2} \Longleftrightarrow \text{Reject } \mathcal{H}_{0},
\end{align*}
where \(T_{1}(\mathbf{y}_{N})\) is the maximum likelihood estimator and \(T_{2}(\mathbf{y}_{N})\) is the Bayesian test statistics based on posterior probability approach, that is, $T_{2}(\mathbf{y}_{N})=\mathbb{P}[\theta > \theta_{0} | \mathbf{y}_{N}]$. $\lambda_{1}$ and $\lambda_{2}$ denote threshold values for the testing procedures. For frequentist testing procedure, we assume that $T_{1}(\mathbf{y}_{N})$ itself serves as the frequentist test statistics of which higher values cast doubt against the null hypothesis \(\mathcal{H}_{0}\), and \(p(\mathbf{y}_{N})\) denotes the p-value. For Bayesian testing procedure, assume that the prior density $\pi(\theta)$ is continuous and positive on the parameter space $\Theta$. 

Under the regularity conditions necessary for the validity of normal asymptotic theory of the maximum likelihood estimator and posterior distribution, and assuming the null hypothesis to be true, it holds that
\begin{align}
\label{eq:asympotoric_eq_Bayesian_Frequentist}
\mathbb{P}[\theta > \theta_{0} | \mathbf{y}_{N}]\approx 1 - p(\mathbf{y}_{N})\quad \text{for large } N,
\end{align}
independently of the form of $\pi(\theta)$.

\end{theorem}
\begin{proof}
Let \(\theta_{t}\) denote the true data generating parameter. Due to the asymptotic normality of the maximum likelihood estimator, if the sample size $N$ is sufficiently large, then the maximum likelihood estimator $T_{1}(\mathbf{y}_{N})$ is normal with mean $\theta_{t}$ and variance $1/I(\theta_{t})$, where \(I(\theta) = \mathbb{E}[(\partial^{2}/\partial \theta^{2}) \log\ f(\mathbf{y}_{N} | \theta)]\) denotes the Fisher information. By setting the performance goal $\theta_{0}$ as the truth \(\theta_{t}\), we can express the p-value of the frequentist test procedure as follow:
\begin{align}
\nonumber
p(\mathbf{y}_{N}) &= \mathbb{P}\left[T_{1}(\mathbf{y}_{N}^{rep}) > T_{1}(\mathbf{y}_{N}) | 
 \mathbf{y}_{N}^{rep}\sim f(y|\theta_{0})\right]\\
 \nonumber
 &=
 \mathbb{P}\left[\sqrt{I(\theta_{0})} \cdot ( T_{1}(\mathbf{y}_{N}^{rep}) - \theta_{0}) >
 \sqrt{I(\theta_{0})} \cdot (  T_{1}(\mathbf{y}_{N}) - \theta_{0}) | 
  \mathbf{y}_{N}^{rep} \sim f(y|\theta_{0})\right]\\
  \nonumber
 &\approx
  \mathbb{P}\left[Z>
 \sqrt{I(\theta_{0})} \cdot (  T_{1}(\mathbf{y}_{N}) - \theta_{0}) | 
  \mathbf{y}_{N}^{rep} \sim f(y|\theta_{0})\right] \quad \text{for large } N\\
  \nonumber
 &=
    \mathbb{P}\left[Z<
 \sqrt{I(\theta_{0})} \cdot ( \theta_{0} -   T_{1}(\mathbf{y}_{N})) | 
 \mathbf{y}_{N}^{rep} \sim f(y|\theta_{0})\right]\\
 \label{eq:p_value}
&= \Phi\left(\sqrt{I(\theta_{0})} \cdot ( \theta_{0} -   T_{1}(\mathbf{y}_{N}))\right),
\end{align}
where $Z$ and $\Phi(x)$ represent the random variable and cumulative distribution function of standard normal distribution, respectively. $\mathbf{y}_{N}^{rep}$ represents a hypothetical representation of data $\mathbf{y}$ under the null distribution.

On the other hand, the Bernstein-Von Mises theorem \citep{johnstone2010high,walker1969asymptotic} states that, if the sample size \(N\) is sufficiently large, the posterior distribution \(\pi(\theta|\mathbf{y}_{N})\) is approximately normally distributed with the mean same as \(T_{1}(\mathbf{y}_{N})\) (the maximum likelihood estimator) and the variance same as the reciprocal of the Fisher information evaluated at the truth, i.e., \(1/I(\theta_{t})\), independently of the form of prior $\pi(\theta)$.

Therefore, the following equation holds 
\begin{align}
\nonumber
T_{2}(\mathbf{y}_{N}) &= 
\mathbb{P}\left[\theta > \theta_{0} | \mathbf{y}_{N}\right]  \\
\nonumber
&=
\mathbb{P}\left[\sqrt{I(\theta_{t})} \cdot (\theta  -  T_{1}(\mathbf{y}_{N}))>
\sqrt{I(\theta_{t})} \cdot ( \theta_0   -  T_{1}(\mathbf{y}_{N}))| \mathbf{y}_{N}\right]  \\
\nonumber
&\approx
\mathbb{P}\left[Z >
\sqrt{I(\theta_{t})} \cdot ( \theta_0   -  T_{1}(\mathbf{y}_{N}))| \mathbf{y}_{N}\right] \quad \text{for large } N \\
\nonumber
&=1 - \Phi\left(\sqrt{I(\theta_{t})} \cdot ( \theta_0   -  T_{1}(\mathbf{y}_{N}))\right),
\end{align}
Therefore, if the true data generating parameter $\theta_{t}$ is the performance goal $\theta_{0}$, it holds
\begin{align}
\label{eq:thruth_Bayesian_test_statistics}
T_{2}(\mathbf{y}_{N}) \approx 1 - \Phi\left(\sqrt{I(\theta_{0})} \cdot ( \theta_0   -  T_{1}(\mathbf{y}_{N}))\right)\quad \text{for large } N.
\end{align}
By the asymptotic equations (\ref{eq:p_value}) and (\ref{eq:thruth_Bayesian_test_statistics}), it holds $T_{2}(\mathbf{y}_{N}) \approx 
1 - p(\mathbf{y}_{N})$ if $N$ is sufficiently large.
\end{proof}

Typically, for regulatory submissions, the significance level of the one-sided superiority test (e.g., $\mathcal{H}_{0}: \theta \leq \theta_{0}$ versus $\mathcal{H}_{A}: \theta > \theta_{0}$, with the performance goal $\theta_{0}$) is $2.5\%$. To achieve a one-sided significance level of $\alpha = 0.025$ for a frequentist design, one would use the decision rule $p(\mathbf{y}_{N}) < 0.025$ to reject the null hypothesis, where $p(\mathbf{y}_{N})$ denotes the p-value. The p-value is often called the `observed significance level' because the value by itself represents the evidence against a null hypothesis based on the observed data $\mathbf{y}_{N}$ \citep{cox2020statistical}.

Theorem \ref{thm:Bayes_Freq_same} states that the value of the Bayesian tail probability (\ref{eq:Bayesian_test_statistics_Post_Prob_Approach}) itself also serves as the evidence for the statistical significance. Furthermore, a Bayesian decision rule of $\mathbb{P}[\theta > \theta_{0} | \mathbf{y}_{N}] > 0.975$ will lead to the one-sided significance level of $0.025$, regardless of the choice of prior, whether it is informative or non-informative, under regularity conditions, if the sample size $N$ is sufficiently large.

\begin{figure}[h!]
\centering
\includegraphics[width=\textwidth]{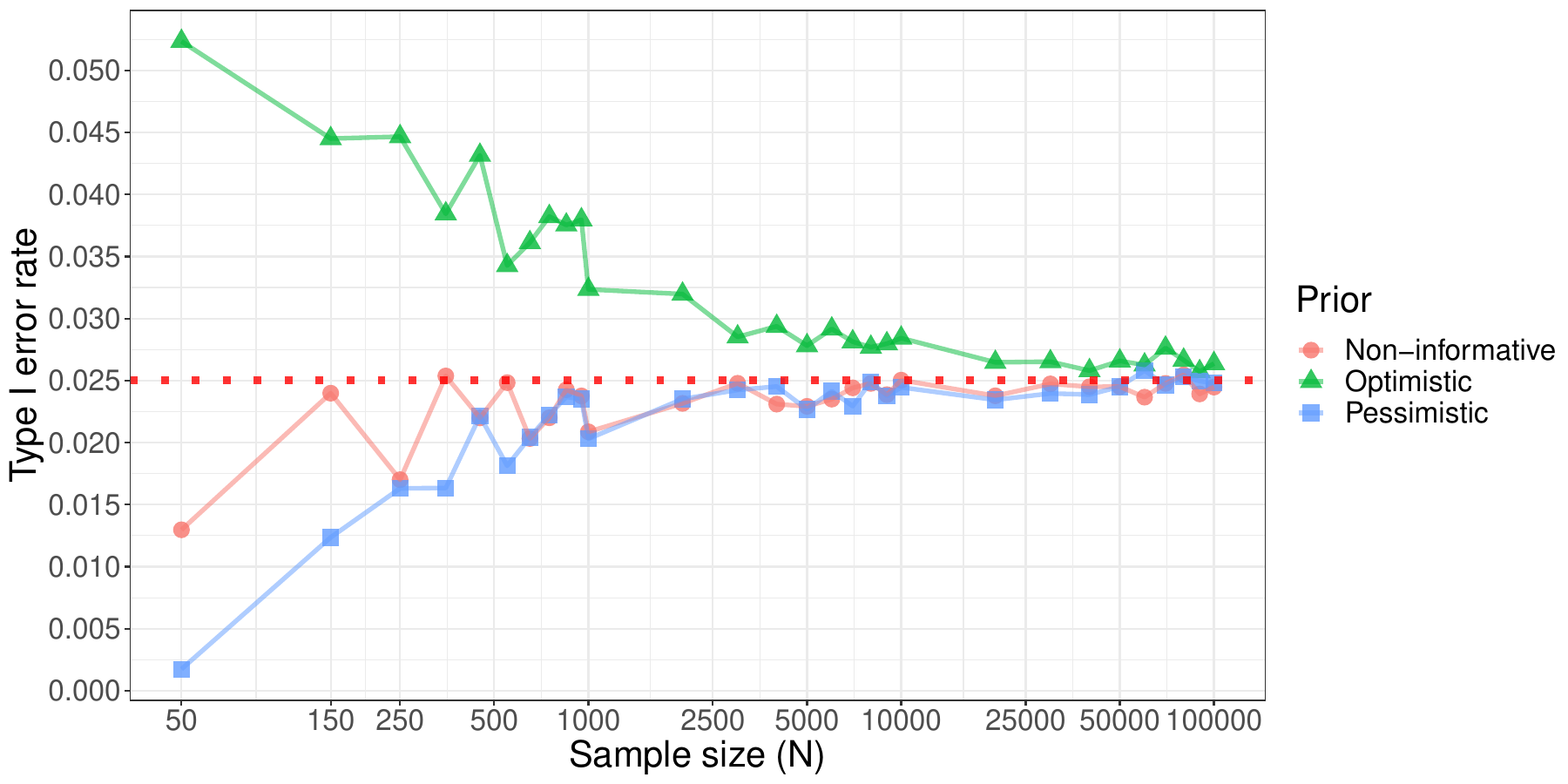}
\caption{\baselineskip=12pt Type I error rates of Bayesian designs based on the beta-binomial model with three prior options for testing $\mathcal{H}_{0}: \theta \geq \theta_{0}$ versus $\mathcal{H}_{a}: \theta < \theta_{0}$, where $\theta_{0} = 0.12$. Prior options are (1) a non-informative prior with $a = b = 1$, (2) an optimistic prior with $a = 0.8$ and $b = 16$, and (3) a pessimistic prior with $a = 3.5$ and $b = 20$.}
\label{fig:Bernstein_Von_Mise_Post_Prob_Approach}
\end{figure}

We illustrate Theorem~\ref{thm:Bayes_Freq_same} by using the beta-binomial model described in Subsection~\ref{subsec:Calibration of Bayesian Trial Design} as an example. Recall that, under sample sizes of $N=100$, $N=150$, and $N=250$, Bayesian designs with non-informative priors meet the type I error requirement, while Bayesian designs with optimistic and pessimistic priors inflate and deflate the type I error, respectively (see Table~\ref{tab:Beta-binomial_single_arm_design}). Under the same settings, we now increase the sample size $N$ up to 100,000 to explore the asymptotic behavior of the Bayesian designs. Figure~\ref{fig:Bernstein_Von_Mise_Post_Prob_Approach} shows the results, where the inflation and deflation induced by the choice of the prior are getting washed out as $N$ increases. When $N$ is as large as 25,000 or more, the type I errors of all the Bayesian designs approximately achieve the type I error rate of 2.5\%, implying that the asymptotic equation~(\ref{eq:Bayesian_test_statistics_Post_Prob_Approach}) holds.

In practice, the sample size ($N$) for pivotal trials in medical device development and phase II trials in drug development often leads to a modest sample size, and there are practical challenges limiting the feasibility of conducting larger studies \citep{faris2017fda}. Consequently, the asymptotic equation (\ref{eq:Bayesian_test_statistics_Post_Prob_Approach}) may not hold in such limited sample sizes. Therefore, sponsors need to conduct extensive simulation experiments in the pre-planning of Bayesian clinical trials to best leverage existing prior information while controlling the type I error rate.
\subsection{Bayesian Group Sequential Design}\label{subsec:Bayesian Group Sequential Design}
An adaptive design is defined as a clinical study design that allows for prospectively planned modifications based on accumulating study data without undermining the study's integrity and validity \citep{Bayesian2010FDAGuidance,ADAPTIVEDRUG2019FDAGuidance,ADAPTIVEMD2016FDAGuidance}. In nearly all situations, to preserve the integrity and validity of a study, modifications should be prospectively planned and described in the clinical study protocol prior to initiation of the study \citep{Bayesian2010FDAGuidance}. Particularly, for Bayesian adaptive designs, including Bayesian group sequential designs, clinical trial simulation is a fundamental tool to explore, compare, and understand the operating characteristics, statistical properties, and adaptive decisions to answer the given research questions \citep{mayer2019simulation}. 

Posterior probability approach is widely adopted as a decision rule for complex innovative designs. In such designs, the choice of the threshold value(s) often depends on several factors, including the complexity of trial design, specific objectives, the presence of interim analyses, ethical considerations, statistical methodology, prior information, and type I \& II error requirements.

Consider a multi-stage design where the sponsor wants to use the posterior probability approach as an early stopping option for the trial success at interim analyses as well as the success at the final analysis. Let $\mathbf{y}^{(k)}$ ($k=1,\ldots,K$) denote the analysis dataset at the $k$-th interim analysis (thus, the $K$-th interim analysis is the final analysis), and $\theta$ is the parameter of main interest. The sponsor wants to test $\mathcal{H}_{0}: \theta \in \Theta_{0}$ versus $\mathcal{H}_{a}: \theta \in \Theta_{a}$, where $\Theta = \Theta_{0} \cup \Theta_{a}$, and $\Theta_{0}$ and $\Theta_{a}$ are disjoint. One can use the following sequential decision criterion:
\begin{align*}
\text{1-st interim analysis} & : T(\textbf{y}^{(1)}) = \mathbb{P}[\theta \in \Theta_{a} | \textbf{y}^{(1)}] > \lambda_{1} \Longleftrightarrow \text{Reject } \mathcal{H}_{0},\\
\text{2-nd interim analysis} & : T(\textbf{y}^{(2)}) = \mathbb{P}[\theta \in \Theta_{a} | \textbf{y}^{(2)}] > \lambda_{2} \Longleftrightarrow \text{Reject } \mathcal{H}_{0},\\
&\quad\quad\quad\quad\quad \vdots\\
\text{K-1-th interim analysis} & : T(\textbf{y}^{(K-1)}) = \mathbb{P}[\theta \in \Theta_{a} | \textbf{y}^{(K-1)}] > \lambda_{K-1} \Longleftrightarrow \text{Reject } \mathcal{H}_{0},\\
\text{K-th interim analysis} & : T(\textbf{y}^{(K)}) = \mathbb{P}[\theta \in \Theta_{a} | \textbf{y}^{(K)}] > \lambda_{K} \Longleftrightarrow \text{Reject } \mathcal{H}_{0},
\end{align*}
Figure \ref{fig:BGSD_general} displays the processes of decision rules based on single-stage design and $K$-stage group sequential design. In practice, a general rule suggests that planning for a maximum of five interim analyses ($K=5$) is often sufficient \citep{pocock2013clinical}. In single-stage design, there is only one opportunity to declare the trial a success. In contrast, sequential design offers $K$ chances to declare success at interim analyses and the final analysis. However, having $K$ opportunities to declare success implies that there are $K$ ways the trial can be falsely considered successful when it's not truly successful. These are the $K$ false positive scenarios, and controlling the overall type I error rate is crucial to maintain scientific integrity for regulatory submission \citep{Bayesian2010FDAGuidance}.

\begin{figure}[h!]
\centering
\includegraphics[width=\textwidth]{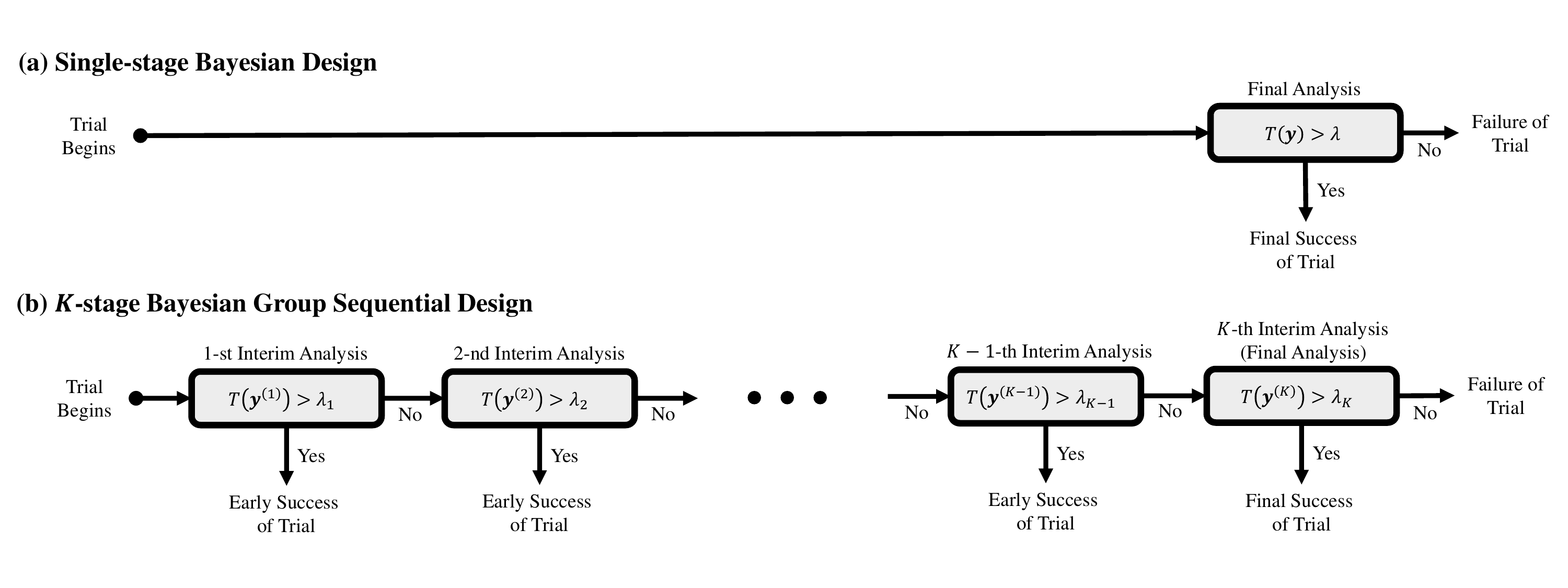}
\caption{\baselineskip=12pt Processes of fixed design (Panel (a)) and sequential design (Panel (b)). The former allows only a single chance to declare success for the trial, while the latter allows $K$ chances to declare success. The test statistic for the former design is denoted as $T(\textbf{y}) = \mathbb{P}[\theta \in \Theta_{a} | \textbf{y}]$, and for the latter design, they are $T(\textbf{y}^{(k)}) = \mathbb{P}[\theta \in \Theta_{a} | \textbf{y}^{(k)}]$, where $(k=1, \cdots, K)$. In both designs, threshold values ($\lambda$ and $\lambda_{k}$, $k=1, \cdots, K$) should be pre-specified before the trial begins to control the type I error rate.}
\label{fig:BGSD_general}
\end{figure}

Similar to frequentist group sequential designs, our primary concern here is to control the overall type I error rate of the sequential testing procedure. The overall type I error rate refers to the probability of falsely rejecting the null hypothesis $\mathcal{H}_{0}$ at any analysis, given that $\mathcal{H}_{0}$ is true. In this example, the overall type I error rate is given by:
\begin{align}
\label{eq:overall_type_I_error_BGSD}
&\mathbb{P}[T(\textbf{y}^{(1)}) > \lambda_{1} \text{ or } \cdots \text{ or } T(\textbf{y}^{(K)}) > \lambda_{K}|\textbf{y}^{(l)}\sim f(y|\theta_{0}), \, (l=1,\cdots,K)]\\
\nonumber
&=\mathbb{P}[T(\textbf{y}^{(1)}) > \lambda_{1} | \textbf{y}^{(1)}\sim f(y|\theta_{0})]\\
\nonumber
&+
\mathbb{P}[T(\textbf{y}^{(1)}) \leq \lambda_{1} \text{ and } T(\textbf{y}^{(2)}) > \lambda_{2}| \textbf{y}^{(l)}\sim f(y|\theta_{0}),\, (l=1,2)]\\
\nonumber
&+
\mathbb{P}[T(\textbf{y}^{(1)}) \leq \lambda_{1} \text{ and } T(\textbf{y}^{(2)}) \leq \lambda_{2} \text{ and } T(\textbf{y}^{(3)}) > \lambda_{3}| \textbf{y}^{(l)}\sim f(y|\theta_{0}),\, (l=1,2,3)]\\
\nonumber
&+\cdots \\
\nonumber
&+
\mathbb{P}[T(\textbf{y}^{(l)}) \leq \lambda_{l}, (l=1,\cdots,K-1) \text{ and }T(\textbf{y}^{(K)}) > \lambda_{K}| \textbf{y}^{(l)}\sim f(y|\theta_{0}), \, (l=1,\cdots,K)],
\end{align}
where $\theta_{0} \in \Theta_{0}$ denotes the null value which leads to the maximum type I error rate (for e.g., $\theta_{0}$ is the performance goal for a single-arm superiority design). Noting from Equation (\ref{eq:overall_type_I_error_BGSD}), the overall type I error rate is a summation of the error rates at each interim analysis. For the relevant calculations corresponding to the frequentist group sequential design, refer to page 10 of \citep{wassmer2016group}, where Bayesian test statistics $T(\mathbf{y}^{(l)})$ and thresholds $\lambda_{l}$ ($l=1,\ldots,K$) are replaced by Z-test statistics based on interim data $\mathbf{y}^{(k)}$ and pre-specified critical values, respectively.

The crucial design objective in the development of a Bayesian group sequential design is to control the overall type I error rate to be less than a significance level of $\alpha$ (typically, 0.025 for a one-sided test and 0.05 for a two-sided test). This objective is similar to what is typically achieved in its frequentist counterparts, such as O’Brien-Fleming \citep{o1979multiple} or Pocock plans \citep{pocock1977group}, or through the alpha-spending approach \citep{demets1994interim}. To achieve this objective, adjustments to the Bayesian thresholds $(\lambda_{1}, \ldots, \lambda_{K})$ are crucial, and this adjustment necessitates extensive simulation work. Failing to make these adjustments may result in an inflation of the overall type I error. For example, if one were to use the same thresholds of $\lambda_{l}=1-\alpha$ ($l=1, \ldots, K$) for all the interim analyses, then the overall type I error would lead to the value greater than $\alpha$ regardless of the maximum number of interim analyses. Furthermore, the overall type I error may eventually converge to $1$ as the number of interim analyses $K$ goes to infinity, similar to the behavior observed in a frequentist group sequential design \citep{armitage1969repeated}. Additionally, compared to single stage designs, group sequential designs may require a larger sample size to achieve the same power all else being equal, as there is an inevitable statistical cost for repeated analyses.

\subsection{Example - Two-stage Group Sequential Design based on Beta-Binomial Model}\label{subsec:Example - Two-stage Group Sequential Design based on Beta-Binomial Model}
We illustrate the advantage of using a Bayesian group sequential design compared to the single-stage Bayesian design described in Subsection \ref{subsec:Calibration of Bayesian Trial Design}. Similar research using frequentist designs can be found in \citep{pocock1982interim}. Recall that the previous design based on a non-informative prior led to a power of 86.90\% and a type I error of 2.31\% with a sample size of \(150\) and a threshold of \(\lambda = 0.975\) (Table \ref{tab:Beta-binomial_single_arm_design}). Our goal here is to convert the fixed-design into a two-stage design that is more powerful, while controlling the overall the type I error rate to be $\alpha \leq 0.025$. For fair comparison, we aim for the expected sample size \(E(N)\) of the two-stage design to be as close to \(150\) as possible. Having a smaller value of \(E(N)\) than \(150\) is even more desirable in our setting because it means that two-stage design can shorten the length of the trial of the fixed design. To compensate for the inevitable statistical cost of repeated analyses, the total sample size of the two-stage design is set to \(N=162\), representing an 8\% increase in the final sample size of the single-stage design. The stage 1 sample size \(N_{1}\) and stage 2 sample size \(N_{2}\) are divided in the ratios of \(3:7\), \(5:5\), or \(7:3\) to see the pattern of probability of early termination with different timing of interim analysis. Finally, we choose \(\lambda_{1} = 0.996\) and \(\lambda_{2} = 0.978\) as the thresholds for the interim analysis and the final analysis, respectively. Note that a more stringent stopping rule has been applied for early interim analyses than for the final analysis, similar to the proposed design of O'Brien and Fleming \citep{o1979multiple}. The same adaptation procedure will be taken to the single-stage designs with final sample sizes of \(100\) and \(200\) as reference.

\begin{table}[h]
\caption{Operating characteristics of two-stage designs based on beta-binomial model.}
		\newcolumntype{C}{>{\centering\arraybackslash}X}
\begin{scriptsize}
\begin{tabular}{cccccccc}
\hline
\begin{tabular}[c]{@{}c@{}}Total \\ Sample Size\\ ($N$)\end{tabular} & \begin{tabular}[c]{@{}c@{}}Stage 1 \\ Sample Size\\ ($N_1$)\end{tabular} & \begin{tabular}[c]{@{}c@{}}Stage 2 \\ Sample Size\\ ($N_2$)\end{tabular} & \begin{tabular}[c]{@{}c@{}}Expected \\ Sample Size\\ ($E(N)$)\end{tabular} & \begin{tabular}[c]{@{}c@{}}Probability of \\ Early Termination\\ ($PET$)\end{tabular} & \begin{tabular}[c]{@{}c@{}}Type I Error\\ ($\alpha$)\end{tabular} & \begin{tabular}[c]{@{}c@{}}Power\\ ($1-\beta$)\end{tabular} & \begin{tabular}[c]{@{}c@{}}\% Change in Power \\ Compared with\\ Single-stage Design\end{tabular} \\ \hline
{\color[HTML]{222222} }                                              & {\color[HTML]{222222} 32}                                                & 76                                                                       & 108                                                                        & {\color[HTML]{222222} 0.0000}                                                         & 0.0199                                                            & 0.7053                                                      & +14.10                                                                                          \\
{\color[HTML]{222222} }                                              & 54                                                                       & 54                                                                       & 105                                                                        & 0.0603                                                                                & 0.0220                                                            & 0.6945                                                      & +12.36                                                                                          \\
\multirow{-3}{*}{{\color[HTML]{222222} 108}}                         & 76                                                                       & 32                                                                       & 100                                                                        & {\color[HTML]{222222} 0.2632}                                                         & {\color[HTML]{222222} 0.0219}                                     & 0.7094                                                      & +14.77                                                                                          \\ \hline
{\color[HTML]{222222} }                                              & 49                                                                       & 113                                                                      & 153                                                                        & {\color[HTML]{222222} 0.0819}                                                         & {\color[HTML]{222222} 0.0200}                                     & 0.8865                                                      & +2.01                                                                                           \\
{\color[HTML]{222222} }                                              & 81                                                                       & 81                                                                       & 145                                                                        & {\color[HTML]{222222} 0.2202}                                                         & {\color[HTML]{222222} 0.0228}                                     & 0.8862                                                      & +1.97                                                                                           \\
\multirow{-3}{*}{{\color[HTML]{222222} 162}}                         & 113                                                                      & 49                                                                       & 146                                                                        & {\color[HTML]{222222} 0.3348}                                                         & {\color[HTML]{222222} 0.0208}                                     & 0.8860                                                      & +1.95                                                                                           \\ \hline
{\color[HTML]{222222} }                                              & 65                                                                       & 151                                                                      & 191                                                                        & 0.1659                                                                                & {\color[HTML]{222222} 0.0219}                                     & 0.9598                                                      & +4.50                                                                                           \\
{\color[HTML]{222222} }                                              & 108                                                                      & 108                                                                      & 177                                                                        & {\color[HTML]{222222} 0.3642}                                                         & {\color[HTML]{222222} 0.0205}                                     & 0.9570                                                      & +4.20                                                                                           \\
\multirow{-3}{*}{{\color[HTML]{222222} 216}}                         & 151                                                                      & 65                                                                       & 183                                                                        & 0.5120                                                                                & {\color[HTML]{222222} 0.0197}                                     & 0.9568                                                      & +4.18                                                                                           \\ \hline
\end{tabular}
\\
\label{tab:Beta-binomial_GSD}
\baselineskip=12pt
Note: All two-stage designs are based on the beta-binomial model with a non-informative prior. The expected sample size (\(E(N)\)) and the probability of early termination (PET) have been calculated under \(\mathcal{H}_{a}\). \(E(N) = N_{1} + (1 - \text{PET}) \cdot N_{2}\), where \(N_{1}\) and \(N_{2}\) denote the sample sizes for stages 1 and 2, respectively. Thresholds for stage 1 and stage 2 are \(\lambda_{1} = 0.996\) and \(\lambda_{2}  = 0.978\), respectively, for all designs. The percentage change in power has been calculated by comparing with the power obtained by the single-stage design (non-informative) in Table \ref{tab:Beta-binomial_single_arm_design}.
\end{scriptsize}
\end{table}

Table \ref{tab:Beta-binomial_GSD} shows the results of the power analysis. It is observed that the overall type I error rates have been protected at 2.5\% for all the considered designs. The expected sample sizes of the two-stage designs using a total sample size of \(N=162\) are \(E(N)=153\) (\(N_{1}:N_{2} = 3:7\)), \(E(N)=145\) (\(N_{1}:N_{2} = 5:5\)), and \(E(N)=146\) (\(N_{1}:N_{2} = 7:3\)), with the power improved from 86.9\% (single-stage design, see Table \ref{tab:Beta-binomial_single_arm_design}) to approximately 88.6\% for all three cases. The power gain is even greater for the two-stage designs using a total sample size of \(N=216\), where the expected sample sizes are smaller than \(N=200\), which is advantageous for using a group-sequential design. Power gains occur for the two-stage designs using a total sample size of \(N=108\) as well, but the expected sample sizes are larger than \(N=100\); therefore, the single-stage design would be preferable in terms of expected sample sizes.

To summarize, the results show that, with an 8\% increase in the final sample size of the single-stage design, we can construct a two-stage design in which the expected sample size is smaller or equal to the final sample size of the single-stage design. This is while still protecting the type I error rate below 2.5\% and benefiting from an increase in the overall power of the designs by as much as 14\% (N=108), 2\% (N=162), and 4\% (N=216), assuming the alternative hypothesis is true. In other words, a group sequential design allowing the claim of early success at interim analysis can help save costs by possibly reducing length of a trial when there is strong evidence of a treatment effect for the new medical device. Even if the evidence is not as strong (null hypothesis seems more likely to be true), the potential risk for the sponsor would be the additional cost spent on enrolling 8\% more patients than with the single-stage design.

\section{Decision Rule - Predictive Probability Approach}\label{sec:Bayesian Decision Rule - Predictive Probability Approach}
\subsection{Predictive Probability Approach}\label{subsec:Predictive Probability Approach}
The primary motivation for employing the predictive probability approach in decision-making is to answer the question at an interim analysis: ``Is the trial likely to present compelling evidence in favor of the alternative hypothesis if we gather additional data, potentially up to the maximum sample size?" This question fundamentally involves predicting the future behavior of patients in the remainder of the study, where the prediction should be based on the interim data observed thus far. Consequently, its idea is akin to measuring conditional power given interim data in the stochastic curtailment method \citep{lachin2005review, gordon1982stochastically}. The key quantity here is the predictive probability of observing a statistically significant treatment effect if the trial were to proceed to its predefined maximum sample size, calculated in a fully Bayesian way. 

One of the most standard applications of predictive probability approach for regulatory submission is the interim analysis for futility stopping (i.e., early stopping the trial in favor of the null hypothesis) \citep{freidlin2002comment, saville2014utility, polack2020safety,snapinn2006assessment}. This is motivated primarily by an ethical imperative; the goal here is to assess whether the trial, based on interim data, is unlikely to demonstrate a significant treatment effect even if it continues to its planned completion. This information can then be utilized by the monitoring committee to assess whether the trial is still viable midway through the trial \citep{demets2016data}. The study will stop for lack of benefit if the predictive probability of success at the final analysis is too small. Other areas where this approach are useful include the early termination for success with consideration of the current sample size (i.e., early stopping the trial in favor of the alternative hypothesis) \citep{wilber2010comparison, lee2008predictive, herson1979predictive}, or sample size re-estimation to evaluate whether the planned sample size is sufficiently large to detect the true treatment effect \citep{broglio2014not}.

We focus on illustrating the use of the predictive probability approach for futility interim analysis. To simplify the discussion, we consider the two-stage futility design where only one interim futility analysis exists. The idea illustrated here can be extended to a multi-stage design by implementing the following testing procedure at each of the interim analyses in the multi-stage design. The logic explained here can be extended to the applications of early success claims and sample size re-estimation after a few modifications.  

Suppose that $\mathbf{y}^{(1)}$ and $\mathbf{y}^{(2)}$ denote the datasets at the interim and final analyses, respectively, and $\theta$ is the main parameter of interest. We distinguish all incremental quantities from cumulative ones using the notation ``tilde." Therefore, $\tilde{\mathbf{y}}^{(2)}$ and $\mathbf{y}^{(2)} = \{\mathbf{y}^{(1)},\tilde{\mathbf{y}}^{(2)} \}$ represent the incremental stage 2 data and the final data, respectively.

At the final analysis, a sponsor plans to test the null hypothesis $\mathcal{H}_{0}: \theta \in \Theta_{0}$ versus the alternative hypothesis $\mathcal{H}_{a}: \theta \in \Theta_{a}$, where $\Theta = \Theta_{0} \cup \Theta_{a}$, and $\Theta_{0}$ and $\Theta_{a}$ are disjoint sets. Suppose that $H(\mathbf{y}^{(2)})$ is the final test statistic to be used, and a higher value casts doubt that the null hypothesis is true. Therefore, the sponsor will claim the success of the trial if it is demonstrated that $H(\mathbf{y}^{(2)}) > \lambda_{2}$ with a predetermined threshold $\lambda_{2}$, where the threshold is chosen to satisfy the type I \& II error requirement of the futility design. It is at the sponsor's discretion whether to use frequentist or Bayesian statistics to construct the final test statistic $H(\mathbf{y}^{(2)})$. This is because the purpose of using the predictive probability approach is to make a decision at the interim analysis, not at the final analysis.

At the interim analysis, the outcomes from stage 1 patients $\mathbf{y}^{(1)}$ are observed. We measure the predictive probability of success at the final analysis, which is the Bayesian test statistics of the predictive probability approach represented as a functional $\mathcal{G}(\cdot): \mathcal{Q}_{\tilde{\mathbf{y}}^{(2)}|\mathbf{y}^{(1)}} \rightarrow [0,1]$, such that:
\begin{align}
\nonumber
\mathcal{G}\{f(\tilde{\mathbf{y}}^{(2)}|\mathbf{y}^{(1)})\} & = T(\mathbf{y}^{(1)}) = \mathbb{P}[H(\mathbf{y}^{(1)},\tilde{\mathbf{y}}^{(2)}) > \lambda_{2}| \mathbf{y}^{(1)} ] \\
\label{eq:Bayesian_test_statistics_Pred_Prob_Approach}
&= \int \mathbf{1}(H(\mathbf{y}^{(1)},\tilde{\mathbf{y}}^{(2)}) > \lambda_{2}) \cdot f(\tilde{\mathbf{y}}^{(2)}|\mathbf{y}^{(1)}) d\tilde{\mathbf{y}}^{(2)},
\end{align}
where $\mathcal{Q}_{\tilde{\mathbf{y}}^{(2)}|\mathbf{y}^{(1)}}$ represents the collection of posterior predictive distributions of stage 2 patient outcome $\tilde{\mathbf{y}}^{(2)}$ given the interim data $\mathbf{y}^{(1)}$. As seen from the integral (\ref{eq:Bayesian_test_statistics_Pred_Prob_Approach}), the fully Bayesian nature of the predictive probability approach is characterized by its integration of final decision results $\mathbf{1}(H(\mathbf{y}^{(1)},\tilde{\mathbf{y}}^{(2)}) > \lambda_{2})$ over the data space of all possible scenarios of future patients' outcome $\tilde{\mathbf{y}}^{(2)}$, with the weight of the integral respecting the posterior predictive distribution $f(\tilde{\mathbf{y}}^{(2)}|\mathbf{y}^{(1)})$. Note that the posterior predictive distribution is again a mixture distribution of the likelihood function of the future outcome $\tilde{\mathbf{y}}^{(2)}$ and the posterior distribution given the interim data:
\begin{align*}
f(\tilde{\mathbf{y}}^{(2)}|\mathbf{y}^{(1)}) &= \int f(\tilde{\mathbf{y}}^{(2)}|\theta) \cdot \pi(\theta|\mathbf{y}^{(1)})d\theta.
\end{align*}


It is important to note that the predictive probability (\ref{eq:Bayesian_test_statistics_Pred_Prob_Approach}) differs from the predictive power \citep{wang2013evaluating,chuang2006sample}, which represents a weighted average of the conditional power, given by $\int \mathbb{P}[H(\mathbf{y}^{(1)}, \tilde{\mathbf{y}}^{(2)}) > \lambda_{2} | \theta] \cdot \pi(\theta | \mathbf{y}^{(1)}) d\theta$. The calculation of the predictive probability (\ref{eq:Bayesian_test_statistics_Pred_Prob_Approach}) follows the fully Bayesian paradigm, while the predictive power a mix of the both frequentist and Bayesian paradigms in the sense that it is constructed based on the conditional power (frequentist statistics) and posterior distribution (Bayesian statistics).

Finally, to induce a dichotomous decision at the interim analysis, we need to pre-specify the futility threshold $\gamma_{1}\in [0,1]$. By introducing an indicator function $\psi$, the testing result for the futility analysis is determined as follow:
\begin{align*}
\psi(\textbf{y}^{(1)}) = \begin{cases}
  1  & \text{ if } \mathcal{G}\{f(\tilde{\mathbf{y}}^{(2)}|\mathbf{y}^{(1)})\} =\mathbb{P}[H(\mathbf{y}^{(1)},\tilde{\mathbf{y}}^{(2)}) > \lambda_{2}| \mathbf{y}^{(1)}] \geq \gamma_{1}\\
  0   & \text{ if } \mathcal{G}\{f(\tilde{\mathbf{y}}^{(2)}|\mathbf{y}^{(1)})\} =\mathbb{P}[H(\mathbf{y}^{(1)},\tilde{\mathbf{y}}^{(2)}) > \lambda_{2}| \mathbf{y}^{(1)}] < \gamma_{1},
\end{cases}
\end{align*}
where $1$ and $0$ indicate the rejection and acceptance of the null hypothesis, respectively. Figure \ref{fig:Pred_Prob_Approach} displays a pictorial description of the decision procedure.

\begin{figure}[h!]
\centering
\includegraphics[scale=0.35]{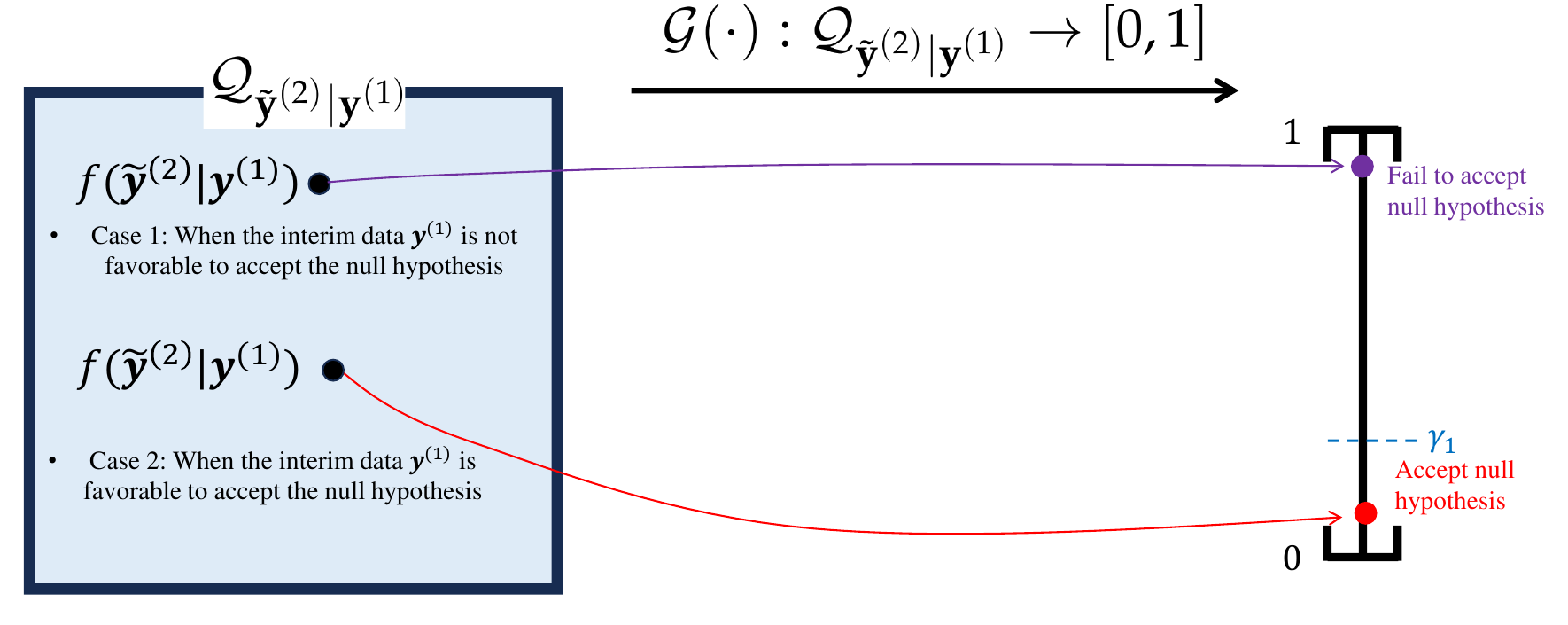}
\caption{\baselineskip=12pt Pictorial illustration of the decision rule based on the predictive probability approach for futility analysis. If the interim data $\mathbf{y}^{(1)}$ favors accepting the null hypothesis, it is also likely that the future remaining patients' outcomes $\tilde{\mathbf{y}}^{(2)}$ would be predicted to be more favorable for accepting the null hypothesis. This prediction results in a lower value of the test statistic $\mathcal{G}\{f(\tilde{\mathbf{y}}^{(2)}|\mathbf{y}^{(1)})\} = \mathbb{P}[H(\tilde{\mathbf{y}}^{(2)},\mathbf{y}^{(1)}) > \lambda_{2}| \mathbf{y}^{(1)}]$ (\ref{eq:Bayesian_test_statistics_Pred_Prob_Approach}). The pre-specified threshold $\gamma_{1}$ is then used to make the dichotomous decision based on the test statistic.}
\label{fig:Pred_Prob_Approach}
\end{figure}

Allowing early termination of a trial for futility tends to reduce the trial’s power as well as the type I error rate \citep{snapinn2006assessment}. To explain this, suppose that one uses the identical final threshold $\lambda_{2}$ in both of the two-stage futility design, as explained above, and the fixed design. Then, the following inequality holds:
\begin{align}
\label{eq:theoretical_fact_of_power_type_I_error_reduction_in_futility_design}
\mathbb{P}_{\theta}[H(\mathbf{y}^{(2)}) > \lambda_{2}]
\geq
\mathbb{P}_{\theta}[T(\mathbf{y}^{(1)}) \geq \gamma_{1} \text{ and } H(\mathbf{y}^{(2)}) > \lambda_{2}], \quad \text{for all }\theta \in \Theta,
\end{align}
which means that the power function of the fixed design is uniformly greater or equal to the power function of the two-stage futility design over the entire parameter space $\Theta$. This implies that equipping a futility rule to a fixed design leads to a reduction of both the type I error rate and power compared to the fixed design.

We briefly discuss the choice of the futility threshold $\gamma_{1}$ and the final threshold $\lambda_{2}$ in the two-stage futility design. Futility threshold $\gamma_{1}$ is typically chosen within the range of 1\% to 20\% in many problems. Having fixed the threshold $\lambda_{2}$, a higher threshold for $\gamma_{1}$ increases the likelihood of discontinuing a trial involving an ineffective treatment, which is desirable because it shortens the trial length when there is a true negative effect. However, it may reduce both the type I error and power compared to a lower threshold for $\gamma_{1}$. On the other hand, the final threshold $\lambda_{2}$ of the futility design is typically chosen to align with the nominal significance level of the corresponding fixed design. This is mainly due to the relevant operational risk of inflating the type I error rate if futility stopping were not executed as planned, even after the final threshold $\lambda_{2}$ has been chosen to make rejection easier to reclaim the lost type I error rate \citep{snapinn2006assessment,lan2003over}. In summary, when constructing a futility design, the sponsor needs to choose the futility threshold that does not substantially affect the operating characteristics of the original fixed-sample size, while also curtailing the trial length when there is a negative effect.

\subsection{Example - Two-stage Futility Design with Greenwood Test}\label{subsec:Example - Two-stage Futility Design with Greenwood Test}
Suppose that a sponsor considers a single-arm design for a phase II trial to assess the efficacy of a new antiarrhythmic drug in treating patients with a mild atrial fibrillation \citep{zimetbaum2012antiarrhythmic}. The primary efficacy endpoint is the freedom from recurrence of the indication at 52 weeks (1 year) after the intervention. The sponsor sets the null and alternative hypotheses by \(\mathcal{H}_{0}: \theta \leq 0.5\) versus \(\mathcal{H}_{a}: \theta > 0.5\), where \(\theta\) denotes the probability of freedom from recurrence at 52 weeks. Let \(S(t)\) represent the survival function; then the main parameter of interest is \(\theta = S(52\text{-week})\). At the planning stage, regulator agreed on the proposal of sponsor that the time to recurrence follows a three-piece exponential model, with a hazard function given as \(h(t) = 0.1 \cdot \xi\) if \(t\in [0, 8\text{-week}]\), \(h(t) = 0.05 \cdot \xi\) if \(t\in (8\text{-week}, 24\text{-week}]\), and \(h(t) = 0.01 \cdot \xi\) if \(t\in (24\text{-week}, 52\text{-week}]\), where \(\xi\) is a positive number. In order to simulate the survival data in the power calculation, the value of \(\xi\) will be derived to set the true data-generating parameter to be \(\theta = S(52\text{-week}) = 0.50, 0.55, 0.60, 0.65,\) and \(0.7\). Note that \(\theta=0.50\) corresponds to the type I error scenario, and the rest of the settings correspond to power scenarios.

We first construct a single-stage design with the final sample size of $N = 100$ patients. The final analysis is conducted by a frequentist hypothesis testing based on the one-sided level-$0.025$ Greenwood test using a confidence interval approach \citep{barber1999symmetric}. More specifically, the testing procedure is that the null hypothesis is rejected if the lower bound of the 95\% two-sided confidence interval evaluated at $t = 52\text{-week}$, that is,
\begin{align}
\label{eq:final_threshold_futility_design}
\text{Study Sucess} = \textbf{1} \left\{LB(\mathbf{y}) = \hat{S}(52\text{-week}) - 1.96 \cdot \sqrt{\hat{Var}[\hat{S}(52\text{-week})]} > 0.5\right\}
\end{align}
Here, the mean estimate $\hat{S}(t)$ is the Kaplan-Meier estimate of $S(t)$ \citep{kaplan1958nonparametric}, and its variance estimate $\hat{Var}[\hat{S}(t)]$ is based on the Greenwood formula \citep{greenwood1926report}, and notation $\mathbf{y}$ represents the final data from $N = 100$ patients. The results of the power analysis obtained by simulation indicate that the probabilities of rejecting the null hypothesis are 0.0185, 0.1344, 0.461, 0.8332, and 0.9793 when the effectiveness success rates ($\theta$) are 0.5, 0.55, 0.60, 0.65, and 0.7, respectively. Note that the type I error rate is 0.0185 less than the 0.025.


Next, we construct a two-stage futility design by equipping the above single-stage design with a non-binding futility stopping option based on the predictive probability approach. Non-binding means that the investigators can freely decide whether they really want to stop or not. This is more common in practice because a stopping decision is typically influenced not only by interim data but also by new external data or safety information \citep{li2020optimality}. The final sample size of the futility design is again $N = 100$, and we keep the final test the same as the single-stage design (\ref{eq:final_threshold_futility_design}). This means that there are no adjustments to the final threshold to reclaim a loss of type I error. The futility analysis will be performed when $N_{1} = 30$ patients have completed the 52 weeks of follow-up (30\% of participants). A non-informative Gamma prior $\mathcal{G}a(0.1, 0.1)$ will be used for each of the hazard rate parameters of the three-piece exponential model. Futility stopping (i.e., accepting the null hypothesis) is triggered if the predictive probability of trial success at the maximum sample size is less than the pre-specified futility threshold $\gamma_{1} = 0.05$. Technically, the predictive probability is
\begin{align*}
T(\mathbf{y}^{(1)}) = \mathbb{P}[LB(\mathbf{y}^{(1)}, \tilde{\mathbf{y}}^{(2)}) > 0.5| \mathbf{y}^{(1)}] =
 \int \mathbf{1}(LB(\mathbf{y}^{(1)}, \tilde{\mathbf{y}}^{(2)}) > 0.5) \cdot f(\tilde{\mathbf{y}}^{(2)}|\mathbf{y}^{(1)}) d\tilde{\mathbf{y}}^{(2)},
\end{align*}
where $\mathbf{y}^{(1)}$ and $\tilde{\mathbf{y}}^{(2)}$ denote the time-to-event outcomes from $N_{1} = 30$ patients and $\tilde{N}_{2} = N - N_{1} = 70$ patients, respectively, and $f(\tilde{\mathbf{y}}^{(2)}|\mathbf{y}^{(1)})$ denotes the posterior predictive distribution of outcomes of the future remaining patients $\tilde{\mathbf{y}}^{(2)}$.


\begin{figure}[h!]
\centering
\includegraphics[width=\textwidth]{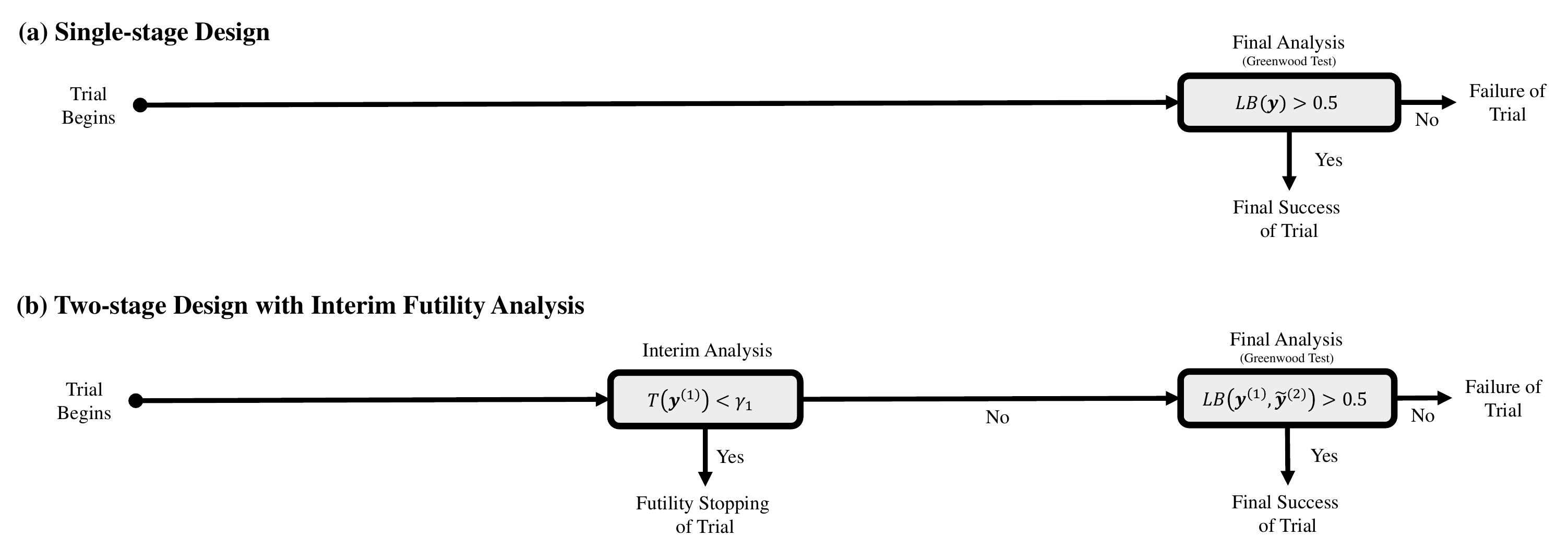}
\caption{\baselineskip=12pt Testing procedures of the single-stage design and the two-stage futility design are as follows: at the final analysis, both designs employ the one-sided level-$0.025$ Greenwood test with a final sample size of $N=100$. Only the futility design has the option to stop the trial due to futility when $N_{1}$ patients had completed 52 weeks of follow-up. In the power analysis, we use $N_{1}=30, 50$, and $70$, along with $\gamma_{1}=0.05$ and $0.1$, to assess the operating characteristics of the design.}
\label{fig:Futility_Design}
\end{figure}

In the power analysis, we vary the number of stage 1 patients, $N_{1}$, to 50 and 70 and set the futility threshold, $\gamma_{1}$, to 0.1 to explore the operating characteristics of the futility design. Figure \ref{fig:Futility_Design} illustrates the testing procedures of the single-stage design and the two-stage futility design. In this setting, the only difference between the futility and single-stage designs is that the former has the option to stop the trial due to futility when $N_{1}$ patients had completed the follow-up of 52 weeks, while the latter does not. Table \ref{tab:Futility Design} shows the power analysis results of the two-stage futility designs.

\begin{table}[h]
\caption{Operating characteristics of two-stage futility designs with the final sample size $N=100$.}
		\newcolumntype{C}{>{\centering\arraybackslash}X}
\begin{tiny}
\begin{tabular}{ccccccccccc}
\hline
                                                                                &                                                                                           & \multicolumn{3}{c}{\begin{tabular}[c]{@{}c@{}}Stage 1 Sample Size\\ $N_1=30$\end{tabular}}                                                                                                                                                                          & \multicolumn{3}{c}{\begin{tabular}[c]{@{}c@{}}Stage 1 Sample Size\\ $N_1=50$\end{tabular}}                                                                                                                                                                          & \multicolumn{3}{c}{\begin{tabular}[c]{@{}c@{}}Stage 1 Sample Size\\ $N_1=70$\end{tabular}}                                                                                                                                                                          \\ \cline{3-11} 
\multirow{-2}{*}{\begin{tabular}[c]{@{}c@{}}Futility \\ Threshold\end{tabular}} & \multirow{-2}{*}{\begin{tabular}[c]{@{}c@{}}Effectiveness\\ Success \\ Rate\end{tabular}} & \begin{tabular}[c]{@{}c@{}}Probability of\\ Rejecting \\ Null \\ Hypothesis\end{tabular} & \begin{tabular}[c]{@{}c@{}}Expected \\ Sample \\ Size\\ ($E(N)$)\end{tabular} & \begin{tabular}[c]{@{}c@{}}Probability of \\ Early \\ Termination\\ ($PET$)\end{tabular} & \begin{tabular}[c]{@{}c@{}}Probability of\\ Rejecting \\ Null \\ Hypothesis\end{tabular} & \begin{tabular}[c]{@{}c@{}}Expected \\ Sample \\ Size\\ ($E(N)$)\end{tabular} & \begin{tabular}[c]{@{}c@{}}Probability of \\ Early \\ Termination\\ ($PET$)\end{tabular} & \begin{tabular}[c]{@{}c@{}}Probability of\\ Rejecting \\ Null \\ Hypothesis\end{tabular} & \begin{tabular}[c]{@{}c@{}}Expected \\ Sample \\ Size\\ ($E(N)$)\end{tabular} & \begin{tabular}[c]{@{}c@{}}Probability of \\ Early \\ Termination\\ ($PET$)\end{tabular} \\ \hline
                                                                                & {\color[HTML]{222222} 0.5}                                                                & {\color[HTML]{222222} 0.017}                                                             & {\color[HTML]{222222} 70.11}                                                  & {\color[HTML]{222222} 0.427}                                                             & {\color[HTML]{222222} 0.017}                                                             & {\color[HTML]{222222} 66.80}                                                   & {\color[HTML]{222222} 0.664}                                                             & {\color[HTML]{222222} 0.018}                                                             & {\color[HTML]{222222} 76.12}                                                  & {\color[HTML]{222222} 0.796}                                                             \\
                                                                                & {\color[HTML]{222222} 0.55}                                                               & {\color[HTML]{222222} 0.116}                                                             & {\color[HTML]{222222} 84.18}                                                  & {\color[HTML]{222222} 0.226}                                                             & {\color[HTML]{222222} 0.117}                                                             & {\color[HTML]{222222} 80.75}                                                  & {\color[HTML]{222222} 0.385}                                                             & {\color[HTML]{222222} 0.12}                                                              & {\color[HTML]{222222} 85.06}                                                  & {\color[HTML]{222222} 0.498}                                                             \\
                                                                                & {\color[HTML]{222222} 0.6}                                                                & {\color[HTML]{222222} 0.441}                                                             & {\color[HTML]{222222} 93.98}                                                  & {\color[HTML]{222222} 0.086}                                                             & {\color[HTML]{222222} 0.439}                                                             & {\color[HTML]{222222} 92.85}                                                  & {\color[HTML]{222222} 0.143}                                                             & {\color[HTML]{222222} 0.446}                                                             & {\color[HTML]{222222} 93.97}                                                  & {\color[HTML]{222222} 0.201}                                                             \\
                                                                                & {\color[HTML]{222222} 0.65}                                                               & {\color[HTML]{222222} 0.818}                                                             & {\color[HTML]{222222} 98.32}                                                  & {\color[HTML]{222222} 0.024}                                                             & {\color[HTML]{222222} 0.817}                                                             & {\color[HTML]{222222} 97.95}                                                  & {\color[HTML]{222222} 0.041}                                                             & {\color[HTML]{222222} 0.829}                                                             & {\color[HTML]{222222} 98.74}                                                  & {\color[HTML]{222222} 0.042}                                                             \\
\multirow{-5}{*}{0.05}                                                          & {\color[HTML]{222222} 0.7}                                                                & {\color[HTML]{222222} 0.975}                                                             & {\color[HTML]{222222} 99.58}                                                  & {\color[HTML]{222222} 0.006}                                                             & {\color[HTML]{222222} 0.973}                                                             & {\color[HTML]{222222} 99.50}                                                   & {\color[HTML]{222222} 0.01}                                                              & {\color[HTML]{222222} 0.976}                                                             & {\color[HTML]{222222} 99.88}                                                  & {\color[HTML]{222222} 0.004}                                                             \\ \hline
                                                                                & {\color[HTML]{222222} 0.5}                                                                & {\color[HTML]{222222} 0.016}                                                             & {\color[HTML]{222222} 59.68}                                                  & {\color[HTML]{222222} 0.576}                                                             & {\color[HTML]{222222} 0.016}                                                             & {\color[HTML]{222222} 63.05}                                                  & {\color[HTML]{222222} 0.739}                                                             & {\color[HTML]{222222} 0.016}                                                             & {\color[HTML]{222222} 74.11}                                                  & {\color[HTML]{222222} 0.863}                                                             \\
                                                                                & {\color[HTML]{222222} 0.55}                                                               & {\color[HTML]{222222} 0.114}                                                             & {\color[HTML]{222222} 76.48}                                                  & {\color[HTML]{222222} 0.336}                                                             & {\color[HTML]{222222} 0.115}                                                             & {\color[HTML]{222222} 76.45}                                                  & {\color[HTML]{222222} 0.471}                                                             & {\color[HTML]{222222} 0.117}                                                             & {\color[HTML]{222222} 82.00}                                                     & {\color[HTML]{222222} 0.600}                                                               \\
                                                                                & {\color[HTML]{222222} 0.6}                                                                & {\color[HTML]{222222} 0.428}                                                             & {\color[HTML]{222222} 89.01}                                                  & {\color[HTML]{222222} 0.157}                                                             & {\color[HTML]{222222} 0.427}                                                             & {\color[HTML]{222222} 89.65}                                                  & {\color[HTML]{222222} 0.207}                                                             & {\color[HTML]{222222} 0.44}                                                              & {\color[HTML]{222222} 92.08}                                                  & {\color[HTML]{222222} 0.264}                                                             \\
                                                                                & {\color[HTML]{222222} 0.65}                                                               & {\color[HTML]{222222} 0.799}                                                             & {\color[HTML]{222222} 95.31}                                                  & {\color[HTML]{222222} 0.067}                                                             & {\color[HTML]{222222} 0.810}                                                              & {\color[HTML]{222222} 96.95}                                                  & {\color[HTML]{222222} 0.061}                                                             & {\color[HTML]{222222} 0.826}                                                             & {\color[HTML]{222222} 97.93}                                                  & {\color[HTML]{222222} 0.069}                                                             \\
\multirow{-5}{*}{0.10}                                                          & {\color[HTML]{222222} 0.7}                                                                & {\color[HTML]{222222} 0.966}                                                             & {\color[HTML]{222222} 98.74}                                                  & {\color[HTML]{222222} 0.018}                                                             & {\color[HTML]{222222} 0.969}                                                             & {\color[HTML]{222222} 99.15}                                                  & {\color[HTML]{222222} 0.017}                                                             & {\color[HTML]{222222} 0.971}                                                             & {\color[HTML]{222222} 99.55}                                                  & {\color[HTML]{222222} 0.015}                                                             \\ \hline
                                                                                & {\color[HTML]{222222} 0.5}                                                                & {\color[HTML]{222222} 0.016}                                                             & {\color[HTML]{222222} 57.93}                                                  & {\color[HTML]{222222} 0.601}                                                             & {\color[HTML]{222222} 0.016}                                                             & {\color[HTML]{222222} 61.20}                                                   & {\color[HTML]{222222} 0.776}                                                             & {\color[HTML]{222222} 0.015}                                                             & {\color[HTML]{222222} 73.15}                                                  & {\color[HTML]{222222} 0.895}                                                             \\
                                                                                & {\color[HTML]{222222} 0.55}                                                               & {\color[HTML]{222222} 0.110}                                                              & {\color[HTML]{222222} 73.61}                                                  & {\color[HTML]{222222} 0.377}                                                             & {\color[HTML]{222222} 0.114}                                                             & {\color[HTML]{222222} 74.40}                                                   & {\color[HTML]{222222} 0.512}                                                             & {\color[HTML]{222222} 0.115}                                                             & {\color[HTML]{222222} 80.56}                                                  & {\color[HTML]{222222} 0.648}                                                             \\
                                                                                & {\color[HTML]{222222} 0.6}                                                                & {\color[HTML]{222222} 0.419}                                                             & {\color[HTML]{222222} 87.12}                                                  & {\color[HTML]{222222} 0.184}                                                             & {\color[HTML]{222222} 0.423}                                                             & {\color[HTML]{222222} 88.20}                                                   & {\color[HTML]{222222} 0.236}                                                             & {\color[HTML]{222222} 0.429}                                                             & {\color[HTML]{222222} 90.67}                                                  & {\color[HTML]{222222} 0.311}                                                             \\
                                                                                & {\color[HTML]{222222} 0.65}                                                               & {\color[HTML]{222222} 0.795}                                                             & {\color[HTML]{222222} 94.47}                                                  & {\color[HTML]{222222} 0.079}                                                             & {\color[HTML]{222222} 0.806}                                                             & {\color[HTML]{222222} 96.60}                                                   & {\color[HTML]{222222} 0.068}                                                             & {\color[HTML]{222222} 0.817}                                                             & {\color[HTML]{222222} 97.27}                                                  & {\color[HTML]{222222} 0.091}                                                             \\
\multirow{-5}{*}{0.15}                                                          & {\color[HTML]{222222} 0.7}                                                                & {\color[HTML]{222222} 0.964}                                                             & {\color[HTML]{222222} 98.6}                                                   & {\color[HTML]{222222} 0.020}                                                              & {\color[HTML]{222222} 0.968}                                                             & {\color[HTML]{222222} 99.05}                                                  & {\color[HTML]{222222} 0.019}                                                             & {\color[HTML]{222222} 0.969}                                                             & {\color[HTML]{222222} 99.31}                                                  & {\color[HTML]{222222} 0.023}                                                             \\ \hline
\end{tabular}
\\
\label{tab:Futility Design}
\baselineskip=12pt
Note: Probabilities of rejecting null hypothesis of single-stage design with the final sample size of $N=100$ are 0.0185, 0.1344, 0.461, 0.8332, and 0.9793 when the effectiveness success rates are 0.5, 0.55, 0.60, 0.65, and 0.7, respectively. 
\end{tiny}
\end{table}

The results demonstrate that the probability of rejecting the null hypothesis in the futility design is consistently lower than that in the single-stage design across various effectiveness success rates ($\theta=0.5, 0.55, 0.6, 0.65,$ and $0.7$). This finding aligns with the inequality (\ref{eq:theoretical_fact_of_power_type_I_error_reduction_in_futility_design}). Specifically, the type I error rates for the futility design are $0.0173$, $0.0160$, and $0.0156$ when the futility thresholds $\gamma_{1}$ are set at $0.05$, $0.10$, and $0.15$, respectively. These results reflect reductions of 6.4\%, 13.5\%, and 15.6\% in the type I error rate compared to the single-stage design. (Recall that the type I error rate of the single-stage design is 0.0185.) This implies that a higher value for the futility threshold $\gamma_{1}$ leads to a more substantial reduction in the type I error rate compared to the single-stage design. A similar pattern of reduction is observed in the power scenarios when $\theta=0.55, 0.6, 0.65,$ and $0.7$.

Notably, the probability of early termination tends to increase as the stage 1 sample size grows from $N_{1}=30$ to $N_{3}=70$. This increase is particularly significant in the type I error scenario when $\theta=0.5$. Across all the scenarios examined, the expected sample size consistently stays below $N=100$. This indicates that the futility design outperforms the single-stage design in terms of expected sample size as a performance criterion. Furthermore, this reduction in expected sample size is even more pronounced in the type I error scenarios. In conclusion, it is evident that for long-term survival endpoints, like the example discussed here, the futility design can lead to substantial resource savings by allowing the trial to be terminated midway when the lack of clinical benefit becomes clear.

\section{Multiplicity Adjustments}\label{sec:Multiplicity Adjustments}
\subsection{Multiplicity Problem - Primary Endpoint Family}\label{subsec:Multiplicity Problem}
Efficacy endpoints are measures designed to reflect the intended effects of a drug or medical device. Clinical trials are often conducted to evaluate the relative efficacy of two or more modes of treatment. For instance, consider a new drug developed for the treatment of heart failure  \citep{rossignol2019heart}. In this case, it may be unclear whether the heart failure drug primarily promotes a decrease in mortality, a reduction in heart failure hospitalization, or an improvement in quality of life. However, demonstrating any of these effects individually would hold clinical significance; there are multiple chances to `win.' Consequently, all three endpoints—mortality rate, number of heart failure hospitalizations, and an index for quality of life—might be designated as separate primary endpoints. This is an illustrative example of a primary endpoint family, and failure to adjust for multiplicity can lead to a false conclusion that the heart failure drug is effective. Here, multiplicity refers to the presence of numerous comparisons within a clinical trial \citep{o1984procedures, dmitrienko2009multiple, dmitrienko2018multiplicity, dmitrienko2013key}. See Section III of the FDA guidance document for the multiple endpoints for more details on the primary endpoint family \citep{guidance2022FDAmultiple}.

In the following, we formulate the multiplicity problem of the primary endpoint family. We consider a family of $K$ primary endpoints, any one of which could support the conclusion that a new treatment has a beneficial effect. For simplicity, we assume that the outcomes of the patients are binary responses, where a response of 1 (yes) indicates that the patient shows a treatment effect. Using the example of a heart failure drug, the first efficacy endpoint measures mortality: whether a patient has survived (yes/no), the second endpoint measures morbidity: whether a patient experienced heart failure hospitalization (yes/no), and the third endpoint measures the quality of life: whether the Kansas City Cardiomyopathy Questionnaire score \citep{spertus2020interpreting} has improved by more than 15 points (yes/no) during a defined period after the treatment. The logic explained in the following can be applied to various types of outcomes, including continuous outcomes and time-to-event outcomes.

We consider a parallel group trial design where subjects are randomized to one of $K$ single arms, each associated with hypotheses given by:
\begin{align}
\label{eq:i_th_hypothesis}
\mathcal{H}_{0,i}: \theta_{i} \leq \theta_{0,i} \quad \text{versus} \quad
\mathcal{H}_{a,i}: \theta_{i} > \theta_{0,i}, \quad (i=1,\cdots,K),
\end{align}
where $\theta_{i}$ denotes the response rate for the $i$-th endpoint (where a higher rate indicates a better treatment effect), and $\theta_{0,i}$ represents the performance goal associated with the $i$-th endpoint.

In a clinical trial with a single endpoint $(K=1)$ tested at $\alpha = 0.025$, the probability of finding a treatment effect by chance alone is at most 0.025. However, multiple testing ($K>1$) can increase the likelihood of type I error (a false conclusion that a new drug is effective). To explain this, suppose that at the final analysis upon completion of the study, the rejection of any one of the null hypotheses among $K$ null hypotheses will lead to marketing approval for a new drug. If there are $K=2$ independent endpoints, each tested at $\alpha = 0.025$, and success on either endpoint by itself would lead to a conclusion of a drug effect, the type I error rate is approximately $5 \approx 1-(1-0.025)^{2}$ percent . With $K=4$ endpoints, the type I error rate increases to about $10 \approx 1-(1-0.025)^{4}$ percent. When there are $K=10$ endpoints, the type I error rate escalates to about $22 \approx 1-(1-0.025)^{10}$ percent. The problem becomes more severe as the number of endpoints ($K$) increases. 

\subsection{Familywise Type I Error Rate and Power}\label{subsec:Type I Error Rate and Power}
It is important to ensure that the evaluation of multiple hypotheses will not lead to inflation of the study's overall type I error probability relative to the planned significance level. This is the primary regulatory concern, and it is required to minimize the chances of a false positive conclusion for any of the endpoints, regardless of which and how many endpoints in the study have no effect \citep{guidance2022FDAmultiple}. This probability of incorrect conclusions is known as the familywise type I error rate \citep{bretz2016multiple}. Technically, it can be written as
\begin{align}
\label{eq:overall_type_I_error}
\alpha^{family} &= \mathbb{P}[\text{Reject at least one null hypothesis}|\text{All null hypotheses are true}]\\
\nonumber
&=\mathbb{P}[\text{Reject the collection } \{\mathcal{H}_{0,i}\}_{i\in A} \text{ for all } A \in \mathcal{K}|\{\mathcal{H}_{0,i}\}_{i=1}^{K} \text{ are true}]\\
\nonumber
&=\mathbb{P}[V \geq 1 |\{\mathcal{H}_{0,i}\}_{i=1}^{K} \text{ are true}].
\end{align}
Here, $P(A)$ and $\emptyset$ denote the power set of set $A$ and the empty set, respectively. If there are $K=4$ endpoints, one needs to consider $15=2^4 - 1$ false positive scenarios, each of which contributes to an increase in $\alpha^{family}$. When $K=10$ endpoints are examined in a study, the number of false positive scenarios increases to $1023=2^{10} - 1$ scenarios. $V$ denotes the number of hypotheses rejected among the $K$ hypotheses, taking an integer value from 0 to K.

Another regulatory concern for a primary endpoint family is to maximize the chances of a true positive conclusion. The desired power is an
important factor in determining the sample size. Unlike the type I error scenario where $\alpha^{family}$ is standardly used in most cases, the concept of power can be generalized in various ways when multiple hypotheses are considered (see Chapter 2 in \citep{bretz2016multiple} for more details). Two following two types of power are frequently used
\begin{align}
\label{eq:disjunctive power}
\pi^{dis} &= \mathbb{P}[\text{Reject at least one null hypothesis}|\text{All null hypotheses are false}]\\
\nonumber
&=\mathbb{P}[\text{Reject the collection } \{\mathcal{H}_{0,i}\}_{i\in A} \text{ for all } A \in \mathcal{K}|\{\mathcal{H}_{0,i}\}_{i=1}^{K} \text{ are false}]\\
\nonumber
&=\mathbb{P}[V \geq 1 | \{\mathcal{H}_{0,i}\}_{i=1}^{K} \text{ are false}],\\
\label{eq:conjunctive power}
\pi^{con} &= \mathbb{P}[\text{Reject all null hypotheses}|\text{All null hypotheses are false}]\\
\nonumber
&=\mathbb{P}[\text{Reject the collection } \{\mathcal{H}_{0,i}\}_{i=1}^{K}|\{\mathcal{H}_{0,i}\}_{i=1}^{K} \text{ are false}]\\
\nonumber
&=\mathbb{P}[V = K | \{\mathcal{H}_{0,i}\}_{i=1}^{K} \text{ are false}].
\end{align}
The former $\pi^{dis}$ (\ref{eq:disjunctive power}) and latter $\pi^{con}$ (\ref{eq:conjunctive power}) are referred to as disjunctive power and conjunctive power, respectively \citep{vickerstaff2019methods}. By definition, the disjunctive power is greater than the conjunctive power if the number of endpoints is more than one ($K =2,3,\cdots$), and both are equal when $K=1$.

Typically, regulators require the study design to have $\alpha^{family} \leq \alpha$, where $\alpha = 0.025$ for a one-sided test and $\alpha = 0.05$ for a two-sided test for a primary endpoint family. On the other hand, study specific discussion is necessary to determine which power (disjunctive power, conjunctive power, or another type) should be used for a given study. For example, if the study's objective is to detect all existing treatment effects, then one may argue that conjunctive power $\pi^{con}$ should be used. However, if the objective is to detect at least one true effect, then disjunctive power $\pi^{dis}$ is recommended \citep{guidance2022FDAmultiple}.

\subsection{Frequentist Method - p Value-based Procedures}\label{subsec:Frequentist Method to Control The Familywise Type I Error}
Much has been written and published on the mathematical aspects of frequentist adjustment procedures for multiple comparisons, and we refer the reader elsewhere for the details \citep{hochberg1987multiple,senn2007power,proschan2000practical}. Here, we briefly explain three popular p-value-based multiplicity adjustment procedures: the Bonferroni, Holm, and Hochberg methods \citep{hochberg1988sharper,holm1979simple}. These methods utilize the p-values from individual tests and can be applied to a wide range of test situations \citep{hommel2011multiple}. The fundamental difference is that the Bonferroni method uses non-ordered p-values, while the Holm and Hochberg methods use ordered p-values. Refer to Section 18 from \citep{kim2021handbook} for excellent summary of these methods.

\paragraph{$\bullet$ Bonferroni Method}
The Bonferroni method is a single-step procedure that is commonly used, perhaps because of its simplicity and broad applicability. It is known that Bonferroni method provides the most conservative multiplicity adjustment \citep{dmitrienko2018multiplicity}. Here, we use the most common form of the Bonferroni method which divides the overall significance level of $\alpha$ (typically 0.025 for the one-sided test) equally among the $K$ endpoints for testing $K$ hypotheses (\ref{eq:i_th_hypothesis}). The method then concludes that a treatment effect is significant at the $\alpha$
level for each one of the $K$ endpoints for which the endpoint's p-value is less than $\alpha/K$. 

\paragraph{$\bullet$ Holm Method}
The Holm procedure is a multi-step step-down procedure. It is less conservative than the Bonferroni method because a success with the smallest p-value allows other endpoints to be tested at larger endpoint-specific alpha levels than does the Bonferroni method. The endpoint p-values resulting from the final analysis are ordered from the smallest to the largest (or equivalently, the most significant to the least significant), denoted as $p_{(1)}\leq \cdots \leq p_{(K)}$. 

We take the following stepwise procedure: (Step 1) the test begins by comparing the smallest p-value, $p_{(1)}$, to $\alpha/K$, the same threshold used in the equally-weighted Bonferroni correction. If this $p_{(1)}$ is less than $\alpha/K$, the treatment effect for the endpoint associated with this p-value is considered significant; (Step 2) the test then compares the next-smallest p-value, $p_{(2)}$, to an endpoint-specific alpha of the total alpha divided by the number of yet-untested endpoints. If $p_{(2)} < \alpha/(K-1)$, then the treatment effect for the endpoint associated with this $p_{(2)}$ is also considered significant; (Step 3) The test then compares the next ordered p-value, $p_{(3)}$, to $\alpha/(K-2)$, and so on until the last p-value (the largest p-value) is compared to $\alpha$; (Step 4) The procedure stops, however, whenever a step yields a non-significant result. Once an ordered p-value is not significant, the remaining larger p-values are not evaluated and it cannot be concluded that a treatment effect is shown for those remaining endpoints.

\paragraph{$\bullet$ Hochberg Method}
The Hochberg procedure is a multi-step step-up testing procedure. It compares the p-values to the same alpha critical values of $\alpha/K, \alpha/(K-1), \cdots,\alpha/2, \alpha$, as the Holm procedure. However, instead of starting with the smallest p-value as performed in Holm procedure, Hochberg procedure starts with the largest p-value (or equivalently, the least significant p-value), which is compared to the largest endpoint-specific critical value $\alpha$. If the first test of hypothesis does not show statistical significance, testing proceeds to compare the second-largest p-value to the second-largest adjusted alpha value, $\alpha/2$. Sequential testing continues in this manner until a p-value for an endpoint is statistically significant, whereupon the Hochberg procedure provides a conclusion of statistically-significant treatment effects for that endpoint and all endpoints with smaller p-values.


\paragraph{$\bullet$ Examples}
For illustration, suppose that a trial with four endpoints $(K=4)$ yielded one-sided p-values of $p_{1}=0.006$ (1-st endpoint), $p_{2}=0.013$ (2-nd endpoint), $p_{3}=0.008$ (3-rd endpoint), and $p_{4}=0.0255$ (4-th endpoint) at the final analysis.

The Bonferroni method compares each of these p-values to $0.00625=0.025/4$, resulting in a significant treatment effect at the 0.025 level for only the 1-st endpoint because only the 1st endpoint has a p-value less than 0.00625.

The Holm method considers the successive endpoint-specific alphas, $0.00625 = 0.025/4$, $0.00833 = 0.025/(4-1)$, $0.0125 = 0.025/(4-2)$, and $0.025= 0.025/(4-3)$. We start by comparing the smallest p-value $p_{1}=0.006$ with 0.00625. The treatment effect for the 1-st endpoint is thus successfully demonstrated, and the test continues to the second step. In the second step, the second smallest p-value is $p_{3} = 0.008$, which is compared to 0.00833. The 3-rd endpoint has, therefore, also successfully demonstrated a treatment effect, as 0.008 is less than 0.00833. Testing can now proceed to the third step, in which the next ordered p-value of $p_{2} = 0.013$ is compared to 0.0125. In this comparison, as 0.013 is greater than 0.0125, the test is not statistically significant. This non significant result stops further tests. Therefore, in this example, the Holm procedure concludes that treatment effects have been shown for the 1st and 3rd endpoints.

The Hochberg method considers the same successive endpoint-specific alphas as the Holm method. In the first step, the largest p-value of $p_{4} = 0.0255$ is compared to its alpha critical value of $\alpha = 0.025$. Because this p-value of 0.0255 is greater than 0.025, the treatment effect for the 4-th endpoint is considered not significant. The procedure continues to the second step. In the second step, the second largest p-value, $p_{2} =0.013$, is compared to $\alpha/2 = 0.0125$. Because $p_{2}$ is greater than the allocated alpha, and the 2-nd endpoint is also not statistically significant, the test continues to the third step. In the third step, the next largest p-value, $p_{3} =0.008$, is compared to its alpha critical value of $\alpha/3 = 0.00833$, and the 3-rd endpoint shows a significant treatment effect. This result automatically causes the treatment effect for all remaining untested endpoints, which have smaller p-values than $0.008$, to be significant as well. Therefore, the 1-st endpoint also shows a significant treatment effect.



\subsection{Bayesian Multiplicity Adjustment Methods}\label{subsec:Bayesian Multiplicity Adjustment Methods}
Bayesian adjustments for multiplicity \citep{lewis2009bayesian,gelman2012we,berry1999bayesian,gopalan1998bayesian} can be acceptable for regulatory submissions, provided the analysis plan is pre-specified and the operating characteristics of the analysis are adequate \citep{Bayesian2010FDAGuidance}. It is advisable to consult regulators early on with regard to a Statistical Analysis Plan that includes Bayesian adjustment for multiplicity.

Generally, the development of Bayesian multiplicity adjustment involves three steps:
\begin{itemize}
\item Step 1: Statistical modeling for the outcomes of endpoints,
\item Step 2: Performing the test for individual hypotheses  (\ref{eq:i_th_hypothesis}) with pre-specified thresholds,
\item Step 3: Interpreting the results of Step 2 in terms of the familywise error rate (\ref{eq:overall_type_I_error}).
\end{itemize}

One of the unique advantages of Bayesian multiplicity adjustment is the flexibility of statistical modeling in the planning phase of Step 1, tailored to the study's objectives, the characteristics of the sub-population, and other relevant factors. For example, if a certain hierarchical or multilevel structure exists among sub-populations (such as, center - doctor - patients as discussed in \citep{zucker1997combining}), then one would use a Bayesian hierarchical model to account for the heterogeneity between sub-populations and patient-to-patient variability simultaneously \citep{stunnenberg2018effect}. Furthermore, adaptive feature can be also incorporated to the Bayesian multiplicity adjustment \citep{berry2013bayesian}. This stands in contrast to traditional frequentist approaches, which evaluate the outcomes from each sub-population independently or simply combine data from all sub-populations through a pooled analysis \citep{mcglothlin2018bayesian}.

In Step 2, sponsors need to provide detailed descriptions of the decision rules that will be used to reject the $i$-th null hypothesis $\mathcal{H}_{0i}$ ($i=1,\cdots,K$) in the Statistical Analysis Plan. The sponsor can choose either the posterior probability approach (Section \ref{sec:Bayesian Decision Rule - Posterior Probability Approach}) or the predictive probability approach (Section \ref{sec:Bayesian Decision Rule - Predictive Probability Approach}) as the decision rules. Most importantly, the threshold value for rejecting each null hypothesis should be pre-specified in the Statistical Analysis Plan, which often requires extensive simulations across all plausible scenarios.

Finally, in Step 3, the results of the $K$ individual tests are interpreted to ensure that the frequentist familywise type I error rate $\alpha^{family}$ (\ref{eq:overall_type_I_error}) is lower than the overall significance level $\alpha$. Additionally, power specific to the study objective (disjunctive power, conjunctive power, or another type) may be measured to estimate the sample size of the study.

\subsection{Bayesian Multiplicity Adjustment using Bayesian Hierarchical Modeling}\label{subsec:Bayesian Hierarchical Modeling}
Here, we illustrate the simplest form of the Bayesian multiplicity adjustment method using Bayesian hierarchical modeling. \citep{thall2003hierarchical, berry2013bayesian, lee2022use, lee2022bayesian}. Bayesian hierarchical modeling is a specific Bayesian methodology that combines results from multiple arms or studies to obtain estimates of safety and effectiveness parameters \citep{neuenschwander2016robust}. This approach is particularly appealing in the regulatory setting when there is an association between the outcomes of $K$ endpoints so that exchangeability of patients' outcomes across $K$ endpoints can be assumed \citep{berry1999bayesian}. Figure \ref{fig:BHM} outlines the three steps of the multiplicity control procedure using a Bayesian hierarchical model.

\begin{figure}[h!]
\centering
\includegraphics[width=\textwidth]{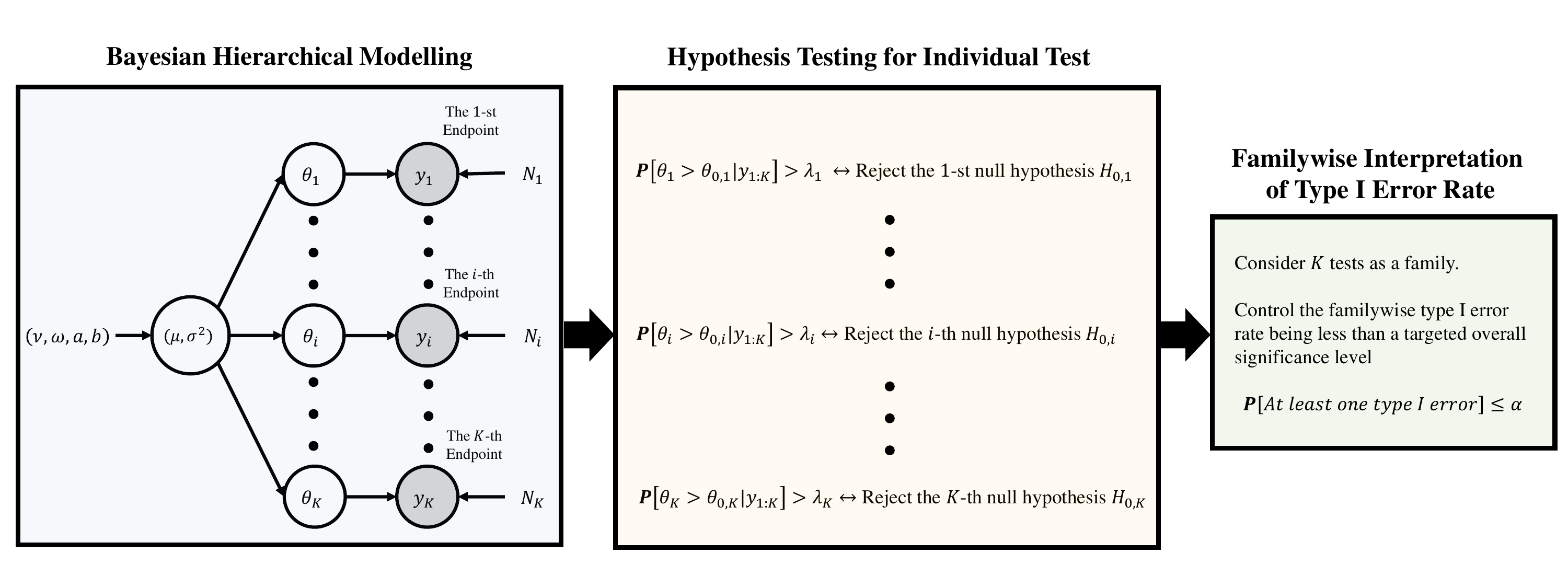}
\caption{\baselineskip=12pt Three steps to control the familywise type I error rate through Bayesian hierarchical modeling. The first step involves specifying a Bayesian hierarchical model, which depends on the context of the problem. In the second step, the decision rule for each individual test is specified. The third step involves interpreting the combination of individual type I error rates in terms of the familywise type I error rate, which is restricted by the overall significance level.}
\label{fig:BHM}
\end{figure}

Let $N_{i}$ be the number of patients to be enrolled in the $i$-th arm associated with the $i$-th endpoint for testing the null and alternative hypotheses, $\mathcal{H}_{0,i}: \theta_{i} \leq \theta_{0,i}$ versus $
\mathcal{H}_{a,i}: \theta_{i} > \theta_{0,i}, \quad (i=1,\cdots,K)$
(\ref{eq:i_th_hypothesis}). The total sample size of the study is therefore $N = \sum_{i=1}^{K} N_{i}$. Let $y_{i}$ denote the number of responders to a treatment, where a higher number indicates better efficacy. Then, the number of responders associated with the $i$-endpoint is distributed according to a binomial distribution:
\begin{align}
\label{eq:BHM_level_1}
y_{i}|\theta_{i}\sim \mathcal{BN}(N_{i},\theta_{i}), \quad (i=1,\cdots,K).
\end{align}
Note that the parameters of main interest are $(\theta_{1},\cdots,\theta_{K})$. Suppose that there is an association between the outcomes of the $K$ endpoints, and $K$ sub-populations are exchangeable, a priori. We assume the most basic formulation of hierarchical prior on the $(\theta_{1},\cdots,\theta_{K})$ given by:
\begin{align}
\label{eq:BHM_level_2}
\phi_{i}|\mu,\sigma^{2} &\sim \mathcal{N}(\mu,\sigma^{2}), \quad
(i=1,\cdots,K),\\
\label{eq:BHM_level_3}
(\mu, \sigma^{2}) &\sim \mathcal{NIG}(\nu,\omega,a,b),
\end{align}
where the parameter $\theta_{i}$ is logit-transformed to $\phi_{i}$ (i.e., $\theta_{i}=\exp(\phi_{i})/\{1 + \exp(\phi_{i})\}$, or equivalently, $\phi_{i} = \log(\theta_{i}/(1-\theta_{i}))$. The normal-inverse-gamma prior, denoted as $(\mu,\sigma^{2}) \sim \mathcal{NIG}(\nu, \omega, a, b)$, is equivalent to a mixture of normal and inverse gamma priors: $\mu|\sigma^{2} \sim \mathcal{N}(\nu,\sigma^{2}/\omega)$ and $\sigma^{2} \sim \mathcal{IG}(a, b)$. $(\nu, \omega, \alpha, \beta)$ represent the hyperparameters, which we set as $(0,$ $1/100,$ $0.001,$ $0.001)$. This choice ensures that the normal-inverse-gamma prior is diffused over the parameter space, and the prior information is almost vague (essentially, nearly non-informative), similar to the choice made by \citep{berry2013bayesian}.

The hierarchical formulation (\ref{eq:BHM_level_1})--(\ref{eq:BHM_level_3}) is designed to induce a shrinkage effect \citep{efron2010future, jones2011bayesian}. Under this formulation, the Bayesian estimators of the parameters $\phi_i, (i=1,\cdots,K)$ (or equivalently, $\theta_i, (i=1,\cdots,K)$) will be pulled toward the global mean $\mu$ (or equivalently, $\exp(\mu)/{1 + \exp(\mu)}$), leading to a reduction in the width of the interval estimates of the parameters, a posteriori, similar to the James-Stein shrinkage estimator \citep{james1992estimation}. This shrinkage effect is also referred to as "borrowing strength", recognized in numerous regulatory guidance documents related to clinical trials for medical
devices and small populations \citep{Small2006EMAGuidance,Bayesian2010FDAGuidance}.

To test the null and alternative hypotheses associated with the $i$-th endpoint (\ref{eq:i_th_hypothesis}), we use the posterior probability approach for decision-making as follow. Upon completion of the study, for each $i$ ($i=1,\cdots,K$), we reject the $i$-th null hypothesis, $\mathcal{H}_{0,i}: \theta_{i} \leq \theta_{0,i}$, if the posterior probability of the $i$-th alternative hypothesis, $\mathcal{H}_{a,i}: \theta_{i} > \theta_{0,i}$, being true is greater than a pre-specified threshold $\lambda_{i}\in [0,1]$. That is, the decision criterion for the $i$-th endpoint is as follow:
\begin{align}
\label{eq:decision_criterion}
\text{Sucess for the } i\text{-th endpoint} = \textbf{1}\{\mathbb{P}[\theta_{i} > \theta_{0,i}| y_{1:K}] > \lambda_{i} \},\quad (i=1,\cdots,K),
\end{align}
where $y_{1:K}$ denotes the numbers of responses from the $K$ endpoints. A higher value of $\lambda_{i}$ leads to a more conservative testing for the $i$-th endpoint, resulting in a lower type I error rate and a lower power, given a fixed sample size $N_{i}$. Posterior probability in (\ref{eq:decision_criterion}) is typically stochastically approximated by using a MCMC method because the posterior distribution, $\pi(\theta_{1:K},\mu,\sigma^{2} | y_{1:K})$, is not represented as a closed-form distribution.

Suppose that the $i$-th null hypothesis has been rejected at the final analysis. In this case, the drug is considered to have demonstrated effects for the $i$-th endpoint. The $K$ threshold values $(\lambda_{1}, \ldots, \lambda_{K})$ in the decision criteria (\ref{eq:decision_criterion}) should be pre-specified during the design stage and chosen through simulation to ensure that the frequentist familywise type I error $\alpha^{family}$ (\ref{eq:overall_type_I_error}) is less than the overall significance level $\alpha$.


\subsection{Simulation Experiment}\label{subse:Simulation Experiment}
We evaluate the performance of Bayesian hierarchical modeling and frequentist methods (specifically, Bonferroni, Holm, and Hochberg procedures) as described in Subsection \ref{subsec:Frequentist Method to Control The Familywise Type I Error} under varying assumptions of the number of endpoints ($K$) from 1 to 10. Regarding the threshold for the decision rule (\ref{eq:decision_criterion}) of Bayesian hierarchical modeling, we use the same value, $\lambda_{i} = 0.985$, for all endpoints $i=1,\cdots,K$, irrespective of the number of endpoints $K$. In other words, there is no specific threshold adjustment concerning the number of endpoints ($K$). 

The thresholds (adjusted alphas) for the Bonferroni, Holm, and Hochberg procedures are described in Subsection \ref{subsec:Frequentist Method to Control The Familywise Type I Error}. Note that the thresholds for the three procedures are set to be increasingly stringent as the number of endpoints $(K)$ increases, aiming to keep the familywise type I error $\alpha^{family}$ less than $\alpha$.

The sample size for each sub-population, $N_{i}$ ($i=1,\cdots,K$), is set to 85 or 100. For a single endpoint $(K=1)$, these sample sizes lead to a power of approximately 80\% ($N_{i}=85$) and 86\% ($N_{i}=100$) based on the Z-test for one proportion at the one-sided significance level $\alpha=0.025$.

The followings are summary of the simulation setting:

\begin{itemize}
\item Number of endpoints: $K = 1,2,\cdots,10$,
\item One-sided significance level: $\alpha = 0.025$,
\item Number of patients: $N_{i} = 85$ or $100,\, (i=1,\cdots,K)$,
\item Performance goals: $\theta_{0,i} = 0.35,\, (i=1,\cdots,K)$,
\item Anticipated rates: $\theta_{a,i} = 0.5,\, (i=1,\cdots,K)$,
\item Multiplicity adjustment methods: 
\begin{enumerate}
\item Bayesian hierarchical modeling (Bayesian),
\item Bonferroni, Holm, and Hochberg procedures (Frequentist),
\end{enumerate}
\item Decision rule:
\begin{enumerate}
\item Bayesian hierarchical modeling: Posterior probability approach (\ref{eq:decision_criterion}) with the threshold $\lambda_{i}=0.985,\, (i=1,\cdots,K)$ across all settings,
\item Bonferroni, Holm, and Hochberg procedures: Use the adjusted p-value as described in Subsection \ref{subsec:Frequentist Method to Control The Familywise Type I Error} such that the unadjusted p-value are obtained by the exact binomial test \citep{clopper1934use}.
\end{enumerate}
\end{itemize}

\begin{figure}[h!]
\centering
\includegraphics[width=\textwidth]{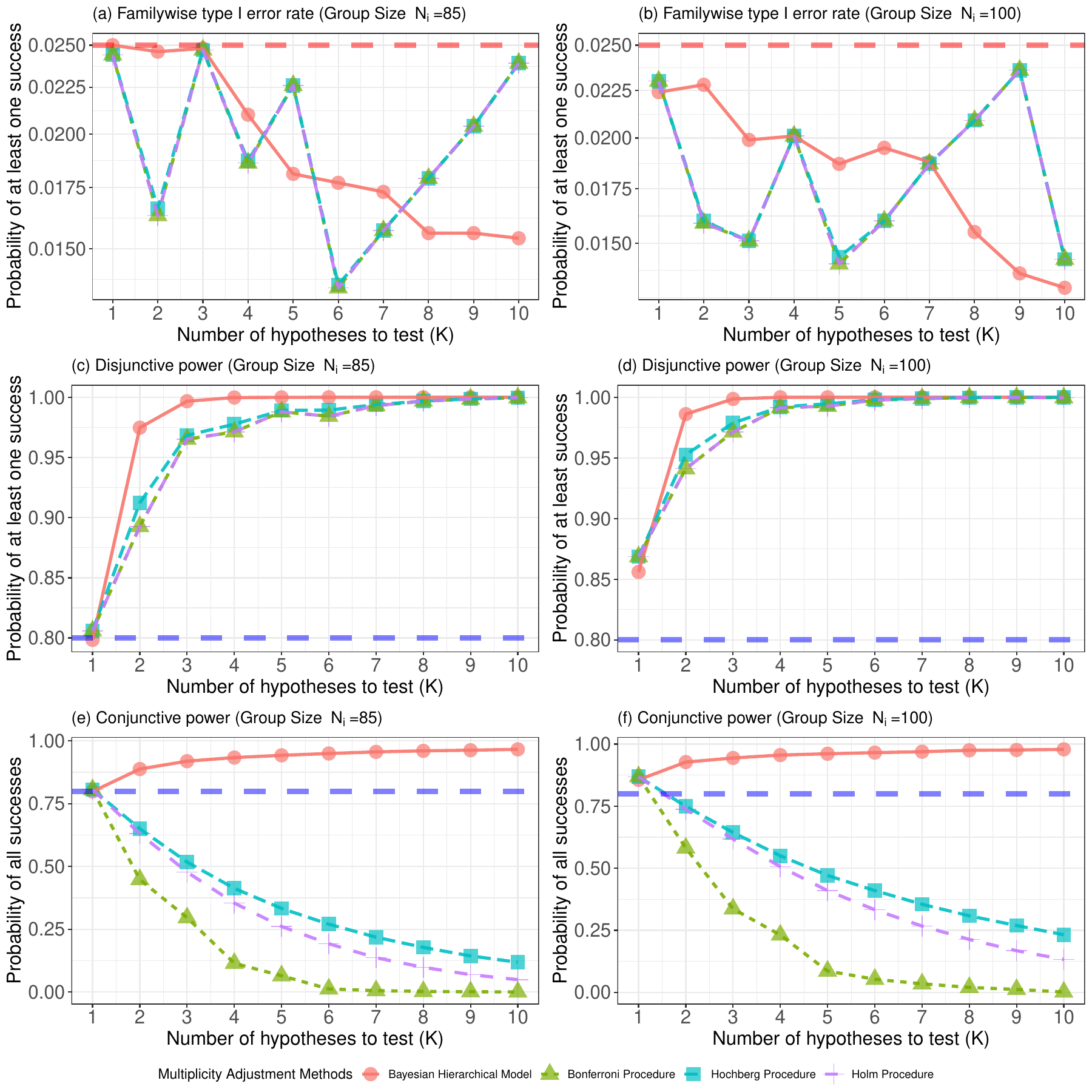}
\caption{\baselineskip=12pt Results of simulation experiment with different number of endpoints ($K=1,\cdots,10$) and group size ($N_{i} = 85, 100$).}
\label{fig:Multiplicity_Problem_Simulation}
\end{figure}

Figure \ref{fig:Multiplicity_Problem_Simulation} displays the results of simulation experiments. Panels (a) and (b) demonstrate that all the considered adjustment methods successfully control the familywise type I error rate, $\alpha^{family}$, at the one-sided significance level of $\alpha=0.025$ across the number of endpoints $K$. Notably, these two panels show that the familywise type I error rate, $\alpha^{family}$, based on Bayesian method decreases as $K$ increases, even when the same thresholds $\lambda_{i}=0.985$ are universally used across all settings. This result implies that there is no need for adjustments of the Bayesian threshold \citep{berry1999bayesian}. Essentially, this nice property is due to the shrinkage effect: borrowing strength across sub-populations automatically adjusts the familywise type I error rate $\alpha^{family}$ to be less than $\alpha=0.025$.

Panels (c) and (d) demonstrate that the disjunctive powers $\pi^{dis}$ (\ref{eq:disjunctive power}) of all the considered adjustment methods increase as $K$ increases. The Bayesian method is the most powerful, while the Bonferroni method is the least powerful among the four methods. The Hochberg method is marginally more powerful than the Holm method. Panels (e) and (f) show that only the Bayesian method leads to an increase in the conjunctive power $\pi^{conj}$ (\ref{eq:conjunctive power}) as $K$ increases. These results indicate that the shrinkage effect of Bayesian hierarchical modeling is beneficial under the two power scenarios. In contrast, p-value-based multiplicity adjustment procedures are only appropriate to use under the disjunctive power scenario. This implies that the total sample size $N=\sum_{i=1}^{K} N_{i}$ required for the study can be significantly reduced if the Bayesian hierarchical model is used, compared to the frequentist methods. Particularly for the conjunctive power scenario, only the Bayesian hierarchical model possesses this unique advantage.

To summarize, the simulation experiment implies that the mechanism of multiplicity adjustment (shrinkage effect or borrowing strength) is automatically embedded in Bayesian hierarchical modeling. This controls the familywise type I error rate to be less than the significance level and improves both disjunctive and conjunctive powers as the number of hypotheses increases. This contrasts with the p-value-based procedures, which are criticized by their overconservatism, which becomes acute when the number of hypotheses is large \citep{hommel2011multiple, simes1986improved, guo2010multiplicity, hochberg1988sharper}.
\section{External Data Borrowing}\label{sec:Historical Data Information Borrowing}
\subsection{Bayesian Information Borrowing for Regulatory Submission}\label{subsec:Bayesian Information Borrowing}
There is a growing interest in Bayesian clinical trial designs with informative prior distributions, allowing the borrowing of information from an external source. Borrowing information from previously completed trials is used extensively in medical device trials \citep{bohm2020efficacy,thompson2021dynamic,Bayesian2010FDAGuidance} and is increasingly seen in drug trials for extrapolation of adult data to pediatrics \citep{gamalo2017statistical} or leveraging historical datasets for rare diseases \citep{gokbuget2016international,gokbuget2016blinatumomab,goring2019characteristics}. In general, sponsors benefit in multiple ways by using Bayesian borrowing designs, including reductions in sample size, time, expense, and increased statistical power. 

In practice, the key difficulty facing stakeholders hoping to design a trial using Bayesian borrowing methods is understanding the similarity of previous studies to the current study, including factors such as enrollment and treatment criteria, and achieving exchangeability between the studies in discussions with regulators. For example, outcomes of medical device trials for a device can vary substantially due to the device evolvement from the previous to the next generation, or by site influenced by differences such as physician training, technique, experience with the device, patient management, and patient population, among many other factors. Regulatory agencies recognize that two studies are never exactly alike; nonetheless, it is recommended that the studies used to construct the informative prior be similar to the current study in terms of the protocol (endpoints, target population, etc.) and the time frame of the data collection to ensure that the practice of medicine and the study populations are comparable \citep{Bayesian2010FDAGuidance}. It is crucial that companies and regulators reach an agreement regarding the prior information and the Bayesian design before enrolling any patients in the new study \citep{campbell2011bayesian}.

One perceptible trend in the Bayesian regulatory environment is that the strict control of the type I error rate in the frequentist framework should be relaxed to a less stringent level for Bayesian submissions using information borrowed from external evidence, due to the unavoidable inflation of the type I error rate \citep{psioda2018bayesian, kopp2020power, psioda2019bayesian, best2023beyond}. In other words, if one wishes to control the type I error rate in the frequentist sense, all prior information must be disregarded in the analysis. Hence, a reasonable control of the type I error is desirable in the setting of information borrowing. Regulators are also increasingly aware of the substantial limitations that stringent control of the frequentist type I error entails. For example, an FDA guidance \citep{Bayesian2010FDAGuidance} states that, `\textit{If the FDA considers the type I error rate of a Bayesian experimental design to be too large, we recommend you modify the design or the model to reduce that rate. Determination of ``too large" is specific to a submission because some sources of type I error inflation (e.g., large amounts of valid prior information) may be more acceptable than others (e.g., inappropriate choice of studies for construction of the prior, inappropriate statistical model, or inappropriate criteria for study success). The seriousness (cost) of a Type I error is also a consideration.}' Several approvals were granted both in the US and in Europe based on non-randomized studies using external controls \citep{goring2019characteristics}. Even though these approvals were typically for rare diseases, they signal the increasing willingness of regulators to review applications for the design and analysis of pivotal clinical studies where the evidence generated outside of a new pivotal study is augmented.

In order to control the type I error rate at a reasonable level with which stakeholders agree, one of the key aspects of Bayesian borrowing designs is to appropriately discount historical/prior information if the prior distribution is too informative relative to the current study \citep{Bayesian2010FDAGuidance}. Although such discounting can be achieved by directly changing the hyper-parameters of the prior, as exemplified by a Beta-Binomial model seen in Table \ref{tab:Beta-binomial_single_arm_design}, or by putting restrictions on the amount of borrowing allowed from previous studies, one of the standard ways is to control the weight parameter on the external study data, which is typically a fractional real number \citep{duan2006evaluating, pawel2023normalized, schmidli2014robust, neuenschwander2009note, ibrahim2015power}. In the next section, we illustrate the use of a power prior model to leverage historical data from a pilot study and explore the influence of the weight parameter on the frequentist operating characteristics of the Bayesian design.

\subsection{Example - Bayesian Borrowing Design based on Power Prior}\label{subsec:Bayesian Borrowing Design based on Power Prior}
We illustrate a Bayesian borrowing design based on a power prior \citep{ibrahim2000power, ibrahim2015power} by taking the primary safety endpoint discussed in Subsection \ref{subsec:Example - Standard Single-stage Design Based on Beta-Binomial Model} as an example. Suppose that a single-arm pilot trial with the number of patients $N_{0}=100$ is done under similar enrollment and treatment criteria as a new pivotal trial. The pilot study provides binary outcome data $\textbf{y}_{N{0}}=(y_{10},\cdots,y_{i0},\cdots,y_{N_{0}0})^{\top}$ for the informative prior in the Bayesian power prior method. The power prior raises the likelihood of the pilot data to the power parameter $a_{0}$, which quantifies the discounting of the pilot data due to heterogeneity between pilot and pivotal trials:
\begin{align}
\label{eq:power_prior}
\pi(\theta|\textbf{y}_{N_{0}},a_{0})&\propto f(\textbf{y}_{N_{0}}|\theta)^{a_{0}} \cdot \pi_{0}(\theta)\propto 
\left\{
\prod_{i=1}^{N_{0}}
\theta^{y_{i0}}
(1 - \theta)^{1 - y_{i0}}
\right\}^{a_{0}} \cdot 
\mathcal{B}\text{eta}(\theta|0.01,0.01)\\
\nonumber
&\propto \mathcal{B}\text{eta}(\theta|a_{0}x_{0} + 0.01,a_{0}(N_{0} - x_{0}) + 0.01),
\end{align}
where $x_{0} = \sum_{i=1}^{N_{0}}y_{i0}$ represents the number of patients who experienced a primary adverse event within 30 days after a surgical procedure involving the device in the pilot trial.

In the power prior formulation (\ref{eq:power_prior}), $\pi_{0}(\theta)$ denotes the prior distribution for $\theta$ before observing the pilot study data $\textbf{y}_{0}$; this is referred to as the initial prior. The initial prior is often chosen to be noninformative, and in this example, we use $\pi_{0}(\theta)=\mathcal{B}\text{eta}(\theta|0.01,0.01).$

The power parameter $a_{0}\in [0,1]$ weighs the pilot data relative to the likelihood of the pivotal trial. The special cases of using the pilot data fully or not at all are covered by $a_{0}=1$ and $a_{0}=0$, respectively, while values of $a_{0}$ between $0$ and $1$ allow for differential weighting of the pilot data. The value $a_{0} N_{0}$ can be interpreted as the prior effective sample size, the number of patients to be borrowed from the pilot study. The parameter $a_{0}$ can be estimated by using the normalized power prior formulation \citep{duan2006evaluating, ye2022normalized}. However, in this paper, we fix $a_{0}$ since our purpose is to explore the influence of the power parameter $a_{0}$ on the frequentist operating characteristics of the Bayesian design.

Finally, the posterior distribution, given the outcomes from patients in pivotal and pilot trials, is once again the beta distribution due to the conjugation relationship:
\begin{align}
\label{eq:Posterior_power_prior}
\pi(\theta|\textbf{y}_{N}, \textbf{y}_{N_{0}},a_{0}) &\propto f(\textbf{y}_{N}|\theta)\cdot \pi(\theta|\textbf{y}_{N_{0}},a_{0})\\
\nonumber
&\propto
\mathcal{B}\text{eta}(\theta|x + a_{0}x_{0} + 0.01,N - x + a_{0}(N_{0} - x_{0}) + 0.01).
\end{align}

Building upon the scenario presented in Subsection \ref{subsec:Example - Standard Single-stage Design Based on Beta-Binomial Model}, the sponsor, during the planning stage of the pivotal trial, anticipated a safety rate of $\theta_{a} = 0.05$ with a performance goal set at $\theta_{0} = 0.12$. At this stage, $\textbf{y}_{N}$ is a random quantity, while $\textbf{y}_{N_{0}}$ is observed, and $a_{0}$ is fixed at a specific value to control the influence of $\textbf{y}_{N{0}}$ in the decision-making process. The decision rule states that if $T(\textbf{y}_{N},\textbf{y}_{N_{0}},a_{0}) = \mathbb{P}[\theta < 0.12 | \textbf{y}_{N},\textbf{y}_{N_{0}},a_{0}] > 0.975$, then the null hypothesis $\mathcal{H}_{0}: \theta \geq 0.12$ is rejected, implying the success of the study in ensuring the safety of the device.

Frequentist operating characteristics of this Bayesian borrowing design can be summarized by two following quantities:
\begin{align}
\label{eq:freq_type_I_error_rate_Bayesian_Borrowing_Design}
&\text{Type I error}: \beta_{\theta_{0}}^{(N)}(\textbf{y}_{N_{0}},a_{0}) = \mathbb{P}[T(\textbf{y}_{N},\textbf{y}_{N_{0}},a_{0})>0.975 |\textbf{y}_{N} \sim f(\textbf{y}_{N}|\theta_{0}), \textbf{y}_{N_{0}},a_{0} ],\\
\label{eq:freq_power_Bayesian_Borrowing_Design}
&\text{Power}: \beta_{\theta_{a}}^{(N)}(\textbf{y}_{N_{0}},a_{0}) = \mathbb{P}[T(\textbf{y}_{N},\textbf{y}_{N_{0}},a_{0})>0.975 | \textbf{y}_{N} \sim f(\textbf{y}_{N}|\theta_{a}), \textbf{y}_{N_{0}},a_{0}].
\end{align}
It is important to note that the type I error rate and power of Bayesian borrowing designs depend on the pilot study data $\textbf{y}_{N_0}$ and the power parameter $a_0$. In the case of no borrowing ($a_0=0$), the values of $\beta_{\theta_{0}}^{(N)}(\textbf{y}_{N_{0}},a_{0})$ (\ref{eq:freq_type_I_error_rate_Bayesian_Borrowing_Design}) and $\beta_{\theta_{a}}^{(N)}(\textbf{y}_{N_{0}},a_{0})$ (\ref{eq:freq_power_Bayesian_Borrowing_Design}) reduce to the values of $\beta_{\theta_{0}}^{(N)}$ (\ref{eq:freq_type_I_error_rate}) and $\beta_{\theta_{a}}^{(N)}$ (\ref{eq:freq_power}), respectively. Otherwise ($0< a_0 \leq 1$), the former values could be significantly different from the latter values.

In the following, we explore the operating characteristics of this Bayesian borrowing design under the two different scenarios regarding the direction of the pilot study data,  whether it is favorable or unfavorable to reject the null hypothesis. In the optimistic external scenario, $x_{0}=5$ out of $N_{0}=100$ patients experienced the adverse event, resulting in a historical event rate of $0.05$, which is lower than the performance goal of $\theta_{0} = 0.12$. In contrast, the pessimistic external scenario is where $x_{0}=15$ out of $N_{0}=100$ patients experienced the adverse event, leading to a historical event rate of $0.15$, which is higher than the performance goal.

Figure \ref{fig:Prior_Prior_Design_Historical_Data_Borrowing} displays the probability of rejecting the null hypothesis versus the power parameter $a_{0}$ for the two scenarios, provided that the sample size for the pivotal trial is $N=150$. The true safety rate $\theta$ is set to be either $\theta_{a} = 0.05$ or $\theta_{0} = 0.12$, corresponding to the power and type I error scenarios, respectively. In the case of no borrowing (that is, $a_{0}=0$), the type I error rate is $0.0225$, and power is $0.8681$, which is almost identical to those obtained from the Bayesian design with a non-informative beta prior and the frequentist design based on z-test statistics seen in Table \ref{tab:Beta-binomial_single_arm_design}.

Panel (a) in Figure \ref{fig:Prior_Prior_Design_Historical_Data_Borrowing} demonstrates that, in the optimistic external scenario, the type I error rate (\ref{eq:freq_type_I_error_rate_Bayesian_Borrowing_Design}) and power (\ref{eq:freq_power_Bayesian_Borrowing_Design}) simultaneously increase as the power parameter $a_{0}$ increases. Conversely, in the pessimistic external scenario (Panel (b)), the type I error rate (\ref{eq:freq_type_I_error_rate_Bayesian_Borrowing_Design}) and power (\ref{eq:freq_power_Bayesian_Borrowing_Design}) simultaneously decrease as the power parameter $a_{0}$ increases. It is important to note that the inflation of the type I error in panel (a) and the deflation of the power in panel (b) are expected (see Subsection \ref{subsec:Example - Standard Single-stage Design Based on Beta-Binomial Model} for relevant discussion).

The central question at this point is, `Is the inflation of the type I error rate (\ref{eq:freq_type_I_error_rate_Bayesian_Borrowing_Design}) under the optimistic scenario scientifically sound for the regulatory submission?' To answer this question, let us assume that the pilot and pivotal studies are very similar and that the pilot study data provide high quality so that the two studies are essentially exchangeable (refer to Subsection 3.7 in \citep{Bayesian2010FDAGuidance} for the concept of exchangeability). Under this idealistic assumption, this inflation is a mathematical result due to the opposite direction of pilot study data $\textbf{y}_{N_{0}}$ (favoring the alternative hypothesis) and pivotal study data $\textbf{y}_{N}$ (generated under the null hypothesis), not due to the incorrect use of the Bayesian borrowing design. Therefore, the inflation of the type I error rate under the optimistic scenario is scientifically sound for the regulatory submission only when the two studies are exchangeable.

\begin{figure}[h!]
\centering
\includegraphics[width=\textwidth]{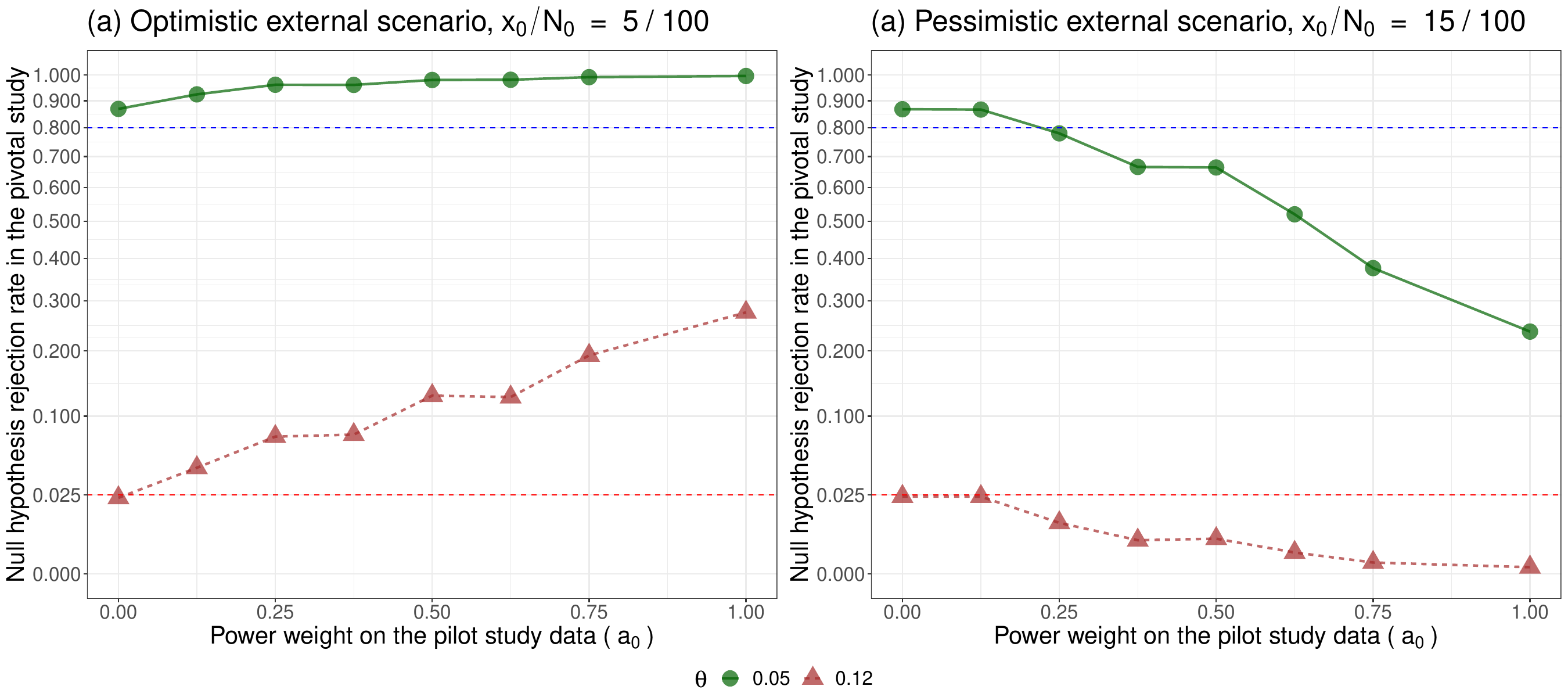}
\caption{\baselineskip=12pt Null hypothesis rejection rate $\beta_{\theta}^{(N)} = \mathbb{P}[T(\textbf{y}_{N},\textbf{y}_{N_{0}},a_{0})>0.975 |\textbf{y}_{N} \sim f(\textbf{y}_{N}|\theta)]$ versus power parameter $a_{0}$ under the optimistic external scenario (Panel (a), $x_{0}/N_{0}=5/100$) and pessimistic external scenario (Panel (b), $x_{0}/N_{0}=15/100$). Sample size of the pivotal trial is $N=150$.}
\label{fig:Prior_Prior_Design_Historical_Data_Borrowing}
\end{figure}

\begin{figure}[h!]
\centering
\includegraphics[width=\textwidth]{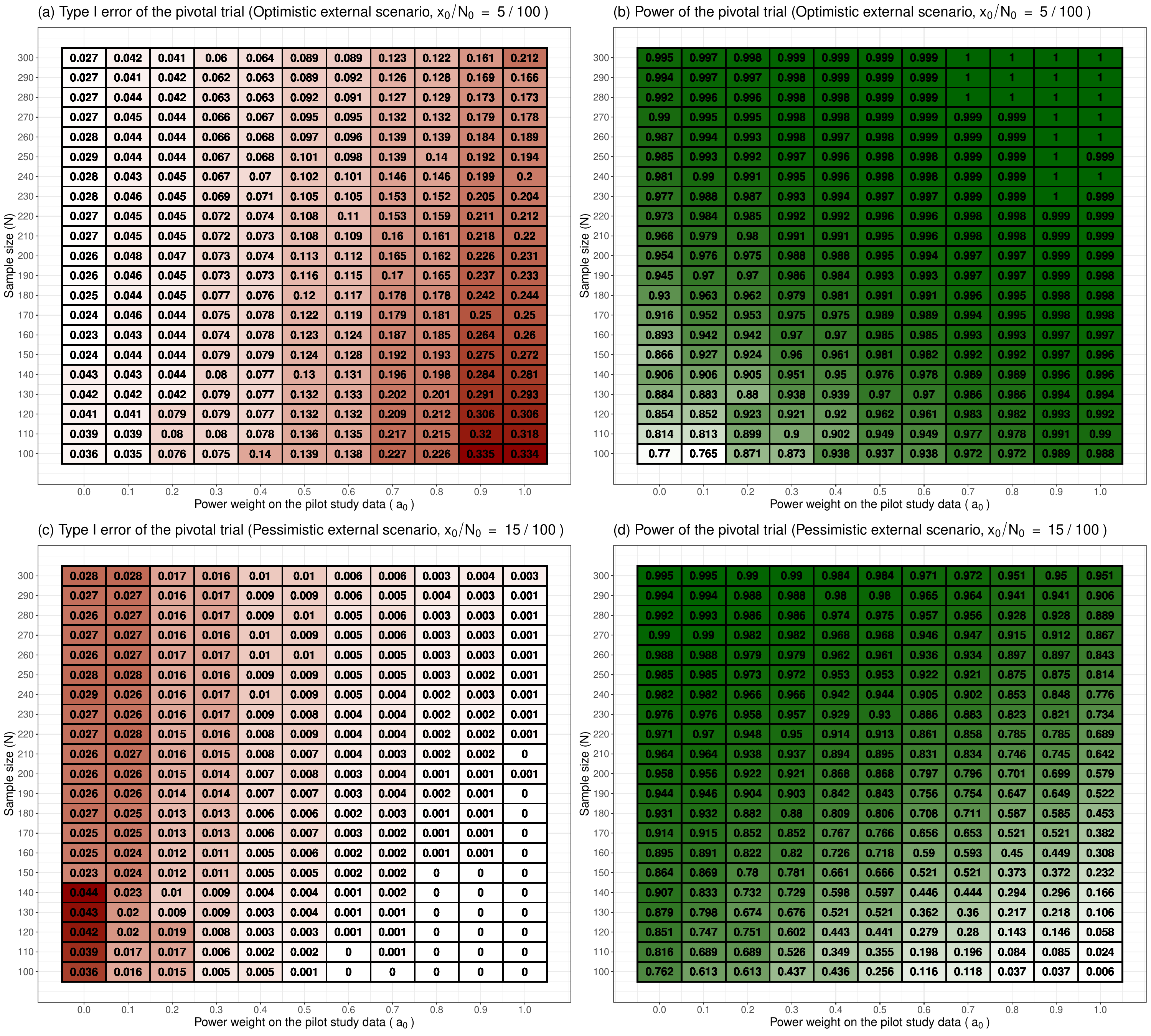}
\caption{\baselineskip=12pt Heatmaps to illustrate the frequentist operating characteristics of the Bayesian borrowing design. The y-axis and x-axis represent the sample size of the pivotal trial $(N)$ and the power parameter $(a_{0})$, respectively. The contents in the heatmaps are the null hypothesis rejection rates $\beta_{\theta}^{(N)} = \mathbb{P}[T(\textbf{y}_{N},\textbf{y}_{N_{0}},a_{0})>0.975 |\textbf{y}_{N} \sim f(\textbf{y}_{N}|\theta)]$, where the type I error rate and power are obtained by setting $\theta = \theta_{0} = 0.12$ and $\theta = \theta_{a} = 0.05$, respectively.}
\label{fig:Power_Prior_Simulation_Results}
\end{figure}

In practice, achieving exchangeability between the two studies is a somewhat elaborate task, and regulators also recognize that no two studies are never exactly alike \citep{Bayesian2010FDAGuidance}. Therefore, the key to the successful submission of a Bayesian borrowing design is to mitigate any potential systematic biases (and consequently the risk of incorrect conclusions) when the pivotal study data $\textbf{y}_{N}$ appears to be inconsistent with the pilot study data $\textbf{y}_{N_{0}}$. This ultimately involves finding an appropriate degree of down-weighting for the pilot study data when such a conflict is present \citep{galwey2017supplementation}. However, this is again a challenging task because, from an operational viewpoint, the pivotal study data $\textbf{y}_{N}$ will be observed upon completion of the study, while the pilot study data $\textbf{y}_{N_{0}}$ has already been observed during the planning phase. The key difficulty here is that the power parameter $a_{0}$ should be determined in the planning phase specified in the protocol or Statistical Analysis Plan before seeing any pivotal study data $\textbf{y}_{N}$. One can estimate the power parameter $a_{0}$ through dynamic borrowing techniques \citep{nikolakopoulos2018dynamic}, but such methods may have their own tuning parameters to control the power parameter $a_{0}$ so the central issue does not completely disappear. 

For this reason, thorough pre-planning is essential when employing Bayesian borrowing designs. This necessitates subject matter expertise, interactions, and a consensus among all stakeholders. It is crucial to establish an agreement on analysis and design priors, with the latter being utilized to assess the operating characteristics of the trial design under all conceivable scenarios. In this regard, a graphical approach can be used to help select design parameters, including the degree of discounting for the pilot study data \citep{edwards2023using}. 

Figure \ref{fig:Power_Prior_Simulation_Results} presents heatmaps for the type I error rate (left heatmaps) and power (right heatmaps) to explore how changing the power parameter ($a_{0}$) and sample size in the pivotal study $(N)$ impacts the type I error and power. As seen from panels (a) and (d), the inflation of the type I error under the optimistic external scenario and the deflation of power under the pessimistic external scenario are evident across the sample size of the pivotal trial ($N$). Another interesting phenomenon is that, as $N$ increases, the tendencies of inflation/deflation diminish across the parameter ($a_{0}$), showcasing the Bernstein-Von Mises phenomenon \citep{johnstone2010high,walker1969asymptotic} as discussed in Subsection \ref{subsec:Asymptotic Property of Posterior Probability Approach}. This suggests that sponsors can benefit from Bayesian borrowing designs in reducing the sample size $N$ only when the pilot study data favorably support rejecting the null hypothesis and $N$ is not excessively large. The acceptable amount of pilot study data to be borrowed should be agreed upon in discussions with regulators because inflation of the type I error rate is expected in this scenario.

\section{Discussion}\label{sec:Discussion}
There have been many Bayesian clinical studies conducted and published in top-tier journals \citep{wilber2010comparison,bohm2020efficacy,richeldi2022trial,baeten2013anti,polack2020safety}. Nevertheless, the adoption of Bayesian statistics for the registration of new drugs and medical devices requires a significant advancement in regulatory science, presenting a range of potential benefits and challenges. In this section, we discuss key aspects of this transformation.

\paragraph{Complex innovative trial designs:} The Bayesian framework provides a promising method to address a variety of modern design complexities as part of complex innovative trial designs. For example, it enables real-time adjustments to trial design, sample size, and patient allocation based on accumulating data from subjects in the trial. These adaptive features can expedite the development of medical products, reduce costs, and enhance patient safety: as exemplified in Subsection \ref{subsec:Example - Two-stage Group Sequential Design based on Beta-Binomial Model} and \ref{subsec:Example - Two-stage Futility Design with Greenwood Test}. Bayesian approach is also frequently useful in the absence of prior information for adaptive clinical trials. However, it is important to note that such modifications are scientifically valid only when prospectively planned and specified in the protocol or Statistical Analysis Plan, and when conducted according to the pre-specified decision rules \citep{bretz2009adaptive,brannath2007multiplicity}. Therefore, it is advisable for sponsors to seek early interaction with regulators regarding the details of their plans for using Bayesian methods \citep{CID2020FDAGuidance}.

\paragraph{Incorporating prior information:} 
One defining feature of Bayesian statistics is the ability to incorporate prior information into the analysis. This contrasts with classical frequentist statistics, which may use information from previous studies only at the design stage. This feature is invaluable when designing clinical trials, especially in situations where historical or more generally study-external data are available. The utilization of informative priors can improve statistical efficiency and enhance the precision of treatment effect estimates. However, it is essential to carefully consider the source and relevance of prior information to ensure the validity and integrity of the trial. Furthermore, as discussed in Section \ref{sec:Historical Data Information Borrowing}, type I error inflation is expected to occur in certain situations. More theoretical work needs to be done in this area to clarify that the stringent control of the type I error probability when there is prior information is not an appropriate way to think about this problem. See Subsection 2.4.3 from \citep{lesaffre2020bayesian} for relevant discussion.

\paragraph{Rare diseases and small sample sizes:} In the context of rare diseases, where limited patient populations hinder traditional frequentist approaches, Bayesian methods are useful. They allow for the integration of diverse data sources, such as historical data or data from similar diseases, to provide robust evidence with a possibly smaller sample size than traditional frequentist approaches. Obtaining ethical and institutional approval is easier in small studies compared with large multicentre studies \citep{hackshaw2008small}. However, as discussed in Subsection \ref{subsec:Asymptotic Property of Posterior Probability Approach}, the operating characteristics of clinical trial designs with a small sample size are more sensitive to the choice of the prior than those with a moderate or large sample size. This implies that smaller clinical trials are more vulnerable to the conflict between the trial data and prior evidence than larger clinical trials. More research is needed in both regulatory science and methodology in this area to mitigate such a conflict and ensure a safe path to regulatory submission, minimizing potential systemic bias.

\paragraph{Regulatory considerations:} The integration of Bayesian statistics into the regulatory setting requires adherence to established guidelines and frameworks. In the past decade, the FDA has recognized the potential of Bayesian approaches and has provided guidance on their use \citep{DesignConsiderations2013FDAGuidance,ADAPTIVEDRUG2019FDAGuidance,Bayesian2010FDAGuidance,ADAPTIVEMD2016FDAGuidance}. However, the adoption of Bayesian statistics is not without challenges and debates. Some statisticians and stakeholders remain cautious about the subjective nature of prior elicitation, potential biases, and the interpretation of Bayesian results. The ongoing debate surrounding the calibration of Bayesian methods, particularly in the context of decision-making, underscores the need for further research and consensus in the field.

In conclusion, the use of Bayesian statistics in clinical trials within the regulatory setting is a promising evolution that can enhance the efficiency and effectiveness of the development process for new drugs or medical devices. However, successful implementation requires rigorous prior specification, careful consideration of decision rules to achieve the study objective, and adherence to regulatory guidelines. The Bayesian paradigm has demonstrated its potential in addressing the complexities of modern clinical trials, offering a versatile tool for researchers and regulators alike. As researchers, clinicians, and regulatory agencies continue to explore the benefits of Bayesian statistics, it is essential to foster collaboration, transparency, and ongoing dialogue to refine and harmonize the use of Bayesian approaches in clinical trials.

\paragraph{Funding details} This research received no external funding.
\paragraph{conflicts of interest} The authors declare no conflict of interest.
\paragraph{Data availability statements} The paper used simulated data.
\paragraph{informed consent} Not applicable.


\bibliographystyle{plain}
\baselineskip=15pt 
\bibliography{MDD_ref}
\end{document}